%% file: main.tex
\documentclass[11pt,a4paper]{article}

\usepackage[a4paper, portrait, margin=2.5cm]{geometry}

\usepackage{amsmath,amssymb,amsthm}
\usepackage{graphicx}
\graphicspath{{./}}
\usepackage{booktabs}
\usepackage[colorlinks=true, citecolor=blue, linkcolor=blue, urlcolor=blue]{hyperref}
\usepackage{authblk}
\usepackage{mathpazo}
\usepackage{microtype}
\usepackage{float}
\usepackage{bm}

\newtheorem{theorem}{Theorem}

\newtheorem{corollary}[theorem]{Corollary}
\newtheorem{proposition}[theorem]{Proposition}
\theoremstyle{remark}
\newtheorem{remark}{Remark}

\begin{document}

\input{dynamic_variables.tex}
\input{dynamic_variables_supplemental.tex}

\title{\LARGE\bfseries A Unified Variational Principle for Branching Transport Networks:\\[4pt]
Wave Impedance, Viscous Flow, and Tissue Metabolism}
\author{Riccardo Marchesi}
\affil{University of Pavia}
\date{\today}

\maketitle

\begin{abstract}
The branching geometry of biological transport networks is canonically
characterized by a diameter scaling exponent $\alpha$. Traditionally, this
exponent interpolates between two structural attractors: impedance matching
($\alpha \sim 2$) for pulsatile wave propagation and viscous-metabolic
minimization ($\alpha=3$) for steady flow. We demonstrate that neither
mechanism in isolation can predict the empirically observed
$\alpha_{\mathrm{exp}}=2.70 \pm 0.20$ in mammalian arterial trees. Incorporating
the empirical sub-linear vessel-wall scaling $h(r) \propto r^p$ ($p=\VarP$) into a
three-term metabolic cost function rigorously breaks the universality of
Murray's cubic law --- a consequence of cost-function inhomogeneity established
via Cauchy's functional equation --- and bounds the static transport optimum to
$\alpha_t \in [\VarAlphaTLow, \VarAlphaTHigh]$. To account for the dynamic
pulsatile environment, we formulate a unified network-level Lagrangian
balancing wave-reflection penalties against steady transport-metabolic costs.
Because the operational duty cycle $\eta$ between pulsatile and steady states is
inherently uncertain over developmental timescales, we cast the morphological
optimization as a zero-sum game between network architecture and environmental
state. By von Neumann's minimax theorem --- for which we provide a direct
constructive proof exploiting the strict monotonicity of the cost curves ---
this game admits a unique saddle point $(\alpha^*, \eta^*)$ satisfying an exact
equal-cost condition, from which the empirical exponent emerges as the robust
optimal compromise between competing thermodynamic demands. We further prove
that $N=2$ (binary branching) uniquely maximizes the network stiffness ratio
$\kappa_{\text{eff}}(N)$, establishing the universal preference for
bifurcations not as an anatomical assumption but as a derived property of the
unified wave-transport framework. Numerical evaluation on the porcine coronary
tree ($G=\VarG$ generations) yields $\alpha^* = \VarAlphaStar$, in quantitative
agreement with morphometric data. Sensitivity analysis confirms that this
prediction is structurally robust to metabolic parameter variation
($|\Delta\alpha^*| < 0.01$ across the physiological range of viscosity and wall
metabolic rates), depending critically only on the histological scaling exponent
$p$ --- the single parameter with direct anatomical grounding. Specifically, the
prediction is analytically insensitive to the exact value of the wall-thickness
pre-factor $c_0$, making the framework a zero-parameter derivation from
fundamental scaling principles.
\end{abstract}

\newpage

\section{Introduction}
\label{sec:intro}

The physical architecture of biological transport networks is canonically
described by Murray's cubic scaling law. However, empirical measurements of
arterial trees consistently yield branching exponents $\alpha \approx
2.7$--$2.9$ \cite{Kassab1993,Zamir1999}. Traditionally, this departure is
attributed either to the structural metabolic cost of the vessel wall or to the
dynamic wave reflection penalties inherent in pulsatile flow. In this Letter, we
demonstrate that neither mechanism in isolation can fully explain the robust
exponent observed in vivo. Instead, we formulate a unified network-level
Lagrangian where the scaling inhomogeneity (incommensurate scaling) of the wall
tissue ($p<1$) drives the system toward a transport-dominated upper bound, while
the geometric impedance mismatch of pulsatile waves tightly constrains the
architecture toward a wave-dominated lower bound. The empirical exponent emerges
strictly as the robust optimal compromise between these competing thermodynamic
demands.
Independently, the engineering literature on acoustic horns \cite{Rayleigh1896},
transmission lines \cite{Pozar2012}, and the hemodynamics of pulsatile flow
\cite{Womersley1955,Nichols2011} establishes that minimizing power reflection at
a junction (impedance matching) requires conservation of total characteristic
impedance.
The matching exponent~$\alpha_w$ depends on the conduit physics: for acoustic
waves in rigid cylinders ($Z \propto d^{-2}$), $\alpha_w = 2$; for electrotonic
conduction in cable-like structures ($Z \propto d^{-3/2}$), $\alpha_w = 3/2$
\cite{Rall1959}; for pulse waves in compliant arteries, $\alpha_w \in [2, 5/2]$
depending on wall-thickness scaling \cite{Nichols2011}.

Empirical measurements report a spectrum of intermediate scaling: mammalian
arterial trees yield $\alpha \approx 2.7$ \cite{Kassab1993, Zamir1999}, while
systems dominated by steady transport, such as pulmonary airways and bronchial
trees, gravitate toward the Murray limit with $\alpha \approx 2.8$--$3.0$
\cite{Weibel1963}. Conversely, networks subject to different impedance-matching
constraints exhibit lower exponents, such as botanical xylem ($\alpha \approx
2.0$--$3.0$ \cite{McCulloh2003, Savage2010}) and cortical dendrites ($\alpha
\approx 2.4$ \cite{Cuntz2010}).
Bennett~\cite{bennett2025} has recently established a unified
variational framework for the \emph{static} branching optimum within the
homogeneous two-term class, deriving generalized Murray scaling,
junction geometry, and a single-index rigidity theorem from a scale-free
cost ledger.
However, this static framework does not incorporate the dynamic
wave-reflection penalties that arise in pulsatile conduits, nor does it
predict the intermediate exponents observed in vivo without a
phenomenological specification of the structural exponent~$m$.

Bejan's constructal law \cite{Bejan1997,Bejan2000} proposes that flow systems
evolve toward configurations that maximize access, offering a broad organizing
principle, and West--Brown--Enquist theory \cite{WBE1997} derives allometric
scaling from fractal filling constraints.
However, neither framework specifies the analytical form of the cost function
that makes wave and flow costs commensurable without introducing a free
parameter.
The neural application connects to Rall's cable theory \cite{Rall1959} and
Cuntz's minimum-spanning-tree models \cite{Cuntz2010}, but again lacks a
thermodynamic justification for the weighting factor between wiring cost and
conduction efficiency.

In this paper, we propose and numerically evaluate a physically derived
network-level framework to address this gap.
We postulate that the effective competition between wave and flow costs is not
determined at a single junction but \emph{emerges at the network level}: wave
reflections accumulate multiplicatively across $G$~generations, while viscous
losses compound through the tree hierarchy.
Both costs become naturally dimensionless at the network scale, eliminating the
need for any dimensional conversion factor, and the branching exponent is
uniquely determined by a minimax (equal-cost) condition requiring that wave and
transport penalties be equal at the network level, eliminating the need for any
phenomenological weighting parameter.
The resulting stiffness ratio $\kappa_\mathrm{eff}(G)$ is an emergent quantity
that grows with tree depth, offering a mechanistic rationale for why deep
networks (arteries, $G \sim \VarG$) produce $\alpha \approx 2.7$ while shallow
networks remain closer to the wave-matching limit.

\section{Theory}
\label{sec:theory}

\subsection{Branching law and impedance scaling}

We consider a symmetric bifurcation ($N = 2$) in which a parent branch of
diameter $d_p$ divides into $N$ daughters of diameter $d_c$:
\begin{equation}
  d_p^\alpha = N \, d_c^\alpha \,.
  \label{eq:branchlaw}
\end{equation}
The binary topology ($N = 2$) is selected by a two-tier argument. At the static
level, the topological bounding theorem of the companion
paper~\cite{Marchesi2026static} proves that the sub-linear wall cost ($p < 1$)
bounds the optimal branching number to small finite integers by penalizing
junction tissue volume. At the dynamic level, Theorem~\ref{thm:topoN} of this
paper proves that among all regular fractal trees perfusing a fixed terminal
population, $N = 2$ uniquely maximizes the network stiffness ratio
$\kappa_\mathrm{eff}(N)$, and consequently uniquely maximizes the minimax
branching exponent $\alpha^*$. The two arguments are complementary: the static
mechanism provides the necessary condition (small $N$), and the dynamic
mechanism provides the sufficient condition (unique optimum). This establishes
binary bifurcation not as an assumption but as a derived property of the unified
wave-transport framework.
The impedance of a cylindrical conduit scales as $Z \propto d^{-\alpha_w}$,
where $\alpha_w$ depends on the wave physics.
Table~\ref{tab:alphaw} summarizes the principal cases.

\begin{table}[H]
\centering
\caption{Impedance scaling exponents for different conduit types.}
\label{tab:alphaw}
\small
\begin{tabular}{@{}lcc@{}}
\toprule
System & $Z$ scaling & $\alpha_w$ \\
\midrule
Rigid tube & $\rho c/A$ & 2 \\
Artery ($h\!\propto\!r$) & $\propto r^{-2}$ & 2 \\
Artery ($h$ const) & $\propto r^{-5/2}$ & 5/2 \\
Artery ($E$ const, $h\!\propto\!r^p$) & $\propto r^{(p-5)/2}$ & $(5{-}p)/2$ \\
Artery ($E\!\propto\!r^{-k_e}$, $h\!\propto\!r^p$) & $\propto r^{(p-k_e-5)/2}$ & $(5{-}p{+}k_e)/2$ \\
Electrotonic cable & $\propto r^{-3/2}$ & 3/2 \\
\bottomrule
\end{tabular}
\end{table}

For arteries with empirically measured wall-thickness scaling $h \propto r^p$
($p \approx \VarP$), the elastic modulus $E$ itself is not strictly constant but
stiffens distally as $E(r) \propto r^{-k_e}$. Empirical porcine coronary data
suggests a gradual stiffening $k_e \approx \VarKe$
\cite{Kassab1993,Nichols2011}. The Moens--Korteweg wave speed thus scales as
$c_\mathrm{MK} \propto (E(r)h/\rho r)^{1/2} \propto r^{(p-1-k_e)/2}$, giving
characteristic impedance $Z \propto \rho c/A \propto r^{(p-5-k_e)/2}$. Perfect
impedance matching dictates the wave attractor to be $\alpha_w = (5-p+k_e)/2$.
For porcine coronary networks, incorporating the empirical elastic stiffening
$k_e \approx \VarKe$ would shift the acoustic ground state from $\alpha_w =
\VarAlphaW$ to $\alpha_w = \VarAlphaWKe$. Because the determination of $k_e$
carries significant order-dependent uncertainty \cite{Kassab1993,Nichols2011},
we retain the baseline $\alpha_w = (5{-}p)/2 = \VarAlphaW$ hereafter to avoid
cascading errors from poorly constrained secondary parameters. The sensitivity
of $\alpha^*$ to this choice is detailed in Table~\ref{tab:sensitivity}.

\subsection{Junction-level cost functions}
\label{sec:junction}

\paragraph{Wave cost.}
A pressure (or voltage) wave encountering the junction sees load impedance
$Z_\mathrm{load} = Z_c/N$.
Using $d_c = d_p \cdot N^{-1/\alpha}$ from Eq.~\eqref{eq:branchlaw}, the
impedance ratio is $Z_\mathrm{load}/Z_p = N^{\alpha_w/\alpha - 1}$, and the
reflected power fraction is \cite{Pozar2012}
\begin{equation}
  |\Gamma(\alpha)|^2 = \left(\frac{N^{\alpha_w/\alpha - 1} - 1}{N^{\alpha_w/\alpha - 1} + 1}\right)^{\!2}.
  \label{eq:gamma_gen}
\end{equation}
This vanishes at $\alpha = \alpha_w$ (perfect matching) with curvature
\begin{equation}
  k_w = \tfrac{1}{2}\,\frac{d^2|\Gamma|^2}{d\alpha^2}\bigg|_{\alpha=\alpha_w}\,.
  \label{eq:kw}
\end{equation}
For $N = 2$, $\alpha_w = \VarAlphaW$: $k_w \approx \VarKwJunctionPhys$.
Table~\ref{tab:gamma} lists exact values.

\begin{table}[H]
\centering
\caption{Reflected power $|\Gamma|^2$ from Eq.~\eqref{eq:gamma_gen} for $N=2$, evaluated at the physiologically relevant $\alpha_w = \VarAlphaW$. The minimum at $\alpha = \alpha_w$ is non-zero only because the tabulated $\alpha$ grid does not coincide exactly with $\alpha_w$.}
\label{tab:gamma}
\small
\begin{tabular}{@{}ccccccc@{}}
\toprule
$\alpha$ & 1.50 & 1.75 & 2.00 & 2.50 & 2.75 & 3.00 \\
\midrule
$|\Gamma|^2$ & \VarGammaAWOneFive{} & \VarGammaAWOneSevenFive{} & \VarGammaAWTwoZero{} & \VarGammaAWTwoFive{} & \VarGammaAWTwoSevenFive{} & \VarGammaAWThreeZero{} \\
\bottomrule
\end{tabular}
\end{table}

\paragraph{Transport cost.}
For steady laminar (Poiseuille) flow through a cylindrical pipe of radius $r$
and length $\ell$ carrying volumetric flow $Q$, the total power expenditure
incorporates both the Rayleigh viscous dissipation of the fluid and the
metabolic maintenance of the biological structure containing the fluid:
\begin{equation}
  \Phi(r, Q, \ell) = \frac{8\mu \ell Q^2}{\pi r^4} + C\pi r^2 \ell + 2\pi\,m_w c_0\, r^{1+p} \ell\,,
  \label{eq:rayleigh}
\end{equation}
where $\mu$ is the dynamic viscosity, $C$ is the metabolic cost per unit blood
volume, $m_w$ is the metabolic rate of wall tissue per unit volume, and $h(r) =
c_0 r^p$ is the wall-thickness. We assume that the vessel-wall metabolism is
dominated by the volume of the smooth muscle cells (SMCs), which leads to a
volumetric cost term $\propto r^{1+p}$. This contrasts with surface-priced
models ($\gamma=1$, $\alpha=2.5$) and highlights the biological cost of the
containing structure.
Murray's classical law ($\alpha_f = 3$) is recovered exactly only in the
non-biological limit where both the structural wall cost is totally negligible
($B_{wall} \to 0$) and the wave penalty is removed ($\eta \to 0$). For
biological systems ($B_{wall} > 0$), as shown by Marchesi
\cite{Marchesi2026static}, the pure-transport limit ceases to be universal. The
inclusion of the sub-linear wall tissue ($p<1$) shatters the cubic geometry and
fundamentally bounds the static flow exponent $\alpha_t$ away from 3.0.

\subsection{The dimensional problem}
\label{sec:dimensional}

At a single junction, the wave cost $|\Gamma|^2$ is a dimensionless fraction of
the incident wave power, while the excess dissipation $\Delta\Phi = \Phi(\alpha)
- \Phi(\alpha_t)$ is measured in watts.
Combining them into a single cost functional requires a dimensional conversion
factor~$\lambda$ (with units of inverse power) whose value cannot be determined
from first principles at the junction level.
This renders the resulting stiffness ratio $\kappa = \lambda\,k_f / k_w$ a free
parameter, limiting predictive power.

We resolve this by formulating the cost functional at the \emph{network level},
where both terms become naturally dimensionless.

\subsection{From junction to network: the emergence of $\kappa$}
\label{sec:network}

Consider a self-similar tree of $G$ generations rooted at a parent vessel of
radius $r_0$ carrying flow $Q_0$, with segment lengths $\ell_g =
\ell_0\,\beta^g$ ($\beta < 1$ for tapering trees).

To avoid circularity and preserve the purely predictive nature of the framework,
the length scaling factor $\beta$ is not treated as a free parameter. Assuming
volumetric isometry where lengths scale linearly with radii ($\ell \propto
r^1$)---consistent with empirical coronary measurements ($\beta \approx 1.0$)
\cite{Kassab1993}---the morphological cascade structurally dictates
\[
  \beta = 2^{-1/\alpha_\mathrm{local}}\,.
\]
With the local static transport optimum $\alpha_{local} \approx \VarAlphaLocal$
determined entirely by the single-vessel optimization of the companion paper
\cite{Marchesi2026static}, this structurally fixes $\beta \approx \VarBeta$.
Rather than treating the length-scaling factor $\beta$ as a free
phenomenological parameter to fit morphometric data, we constrain it strictly
via this theoretical static ground state. Crucially, the volumetric isometry
assumption $\ell \propto r$ is independently supported by coronary morphometry
\cite{Kassab1993} without reference to any branching-exponent measurement:
$\beta$ is therefore an empirically grounded geometric constraint, not a
parameter tuned to reproduce $\alpha^*$.
For completeness, if volumetric isometry is relaxed to
$\ell \propto r^\gamma$ with $\gamma \neq 1$, the length-scaling
factor becomes $\beta = 2^{-\gamma/\alpha_{\mathrm{local}}}$; for
the empirically observed range $\gamma \in [0.9, 1.1]$~\cite{Kassab1993},
this shifts $\beta$ by less than $3\%$ and displaces $\alpha^*$ by
$|\Delta\alpha^*| < 0.02$, well within the empirical uncertainty.
This enforces strict geometric isometry on the unperturbed architectural
scaffold, ensuring the dynamic framework remains rigorously zero-parameter. If
$\beta$ were left as a free variable or dynamically coupled, it would introduce
an unconstrained degree of freedom, weakening the predictive power of the
variational principle.

\paragraph{Network wave cost.}
At each of the $G$ junctions, a fraction $|\Gamma(\alpha)|^2$ of the remaining
wave power is reflected.
Assuming multiplicative independence of reflections at successive junctions (no
multiple scattering), the total transmitted fraction is $(1-|\Gamma|^2)^G$,
giving the network wave cost
\begin{equation}
  \mathcal{C}_\mathrm{wave}^\mathrm{net}(\alpha) = 1 - \bigl(1 - |\Gamma(\alpha)|^2\bigr)^G\,.
  \label{eq:Cwavenet}
\end{equation}
Rather than an approximation of coherent wave propagation, the multiplicative
formula $\mathcal{C}_{\mathrm{wave}}^{\mathrm{net}} = 1 - (1 - |\Gamma|^2)^G$ is
the exact operator that isolates the topological branching penalty from
terminal-load reflections---the quantity directly controlled by network
architecture and therefore the correct object to optimize. It treats junctions
as incoherent geometric scatterers, thereby decoupling the architectural penalty
from frequency-specific phase effects: while strict phase incoherence in
physical transmission lines typically requires $kL \gg 1$, employing the
incoherent sum here rigorously isolates the topological impedance penalty of the
network architecture from frequency-specific phase cancellations. In the actual
coronary regime ($kL \ll 1$), phase coherence is maximal, leading to coherent
signal amplification. A full coherent transfer-matrix calculation incorporating
Womersley viscous corrections and RC terminal loads confirms that this exact
coherent penalty scales proportionally to our incoherent topological metric
(Pearson correlation $R = \VarCoherentCorr$), with a negligible minimax shift
$|\Delta\alpha^*| < 0.01$ (Supplemental~S1). The incoherent formulation
therefore captures the exact geometric scaling behavior while maintaining
analytical tractability. Quantitatively, incorporating a full Womersley viscous
correction across the physiological range $Wo \in [1, 10]$ modifies the
wave-dominated lower bound by less than $\mathcal{O}(10^{-2})$, rendering the
topological result structurally robust to pulsatile velocity profiles.

\paragraph{Network transport cost (two-level formulation).}
At each generation~$g$, the active peak flow $Q_g = Q_{\mathrm{peak}}/2^g$
determines a locally optimal radius $r^*(Q_g)$ via the single-junction
optimization of the companion paper~\cite{Marchesi2026static} (including wall
cost). The resulting locally optimal radii define the transport ground state of
the network, governed by a local morphological exponent
$\alpha_\mathrm{local}(Q)$. Because $\alpha_\mathrm{local}$ is a function of
flow, treating the network transport optimum $\alpha_t$ as a global constant is
technically in tension with the non-universality result established in the
companion paper~\cite{Marchesi2026static}.
However, calculation over the $\VarG$ generations of the porcine coronary tree
reveals that $\alpha_\mathrm{local}$ remains essentially constant, spanning only
$[\VarAlphaLocalMin, \VarAlphaLocalMax]$ across four decades of flow. This
$\mathcal{O}(10^{-2})$ spread biologically justifies the application of a
uniform scaling penalty. Imposing a uniform branching exponent~$\alpha$ on the
entire tree forces radii $r_g = r^*(Q_0) \cdot 2^{-g/\alpha}$ that may deviate
from the locally optimal $r^*(Q_g)$. Writing $r^*_g \equiv r^*(Q_g)$ and
omitting unchanged arguments $Q_g, \ell_g$ for brevity, the transport cost
measures this penalty:
\begin{equation}
  \mathcal{C}_\mathrm{transport}^\mathrm{net}(\alpha) =
  \frac{\displaystyle\sum_{g=0}^{G-1} 2^g \left[ \Phi(r_g) - \Phi(r^*_g) \right]}{\displaystyle\sum_{g=0}^{G-1} 2^g\,\Phi(r^*_g)}\,,
  \label{eq:Ctransportnet}
\end{equation}
which is dimensionless and vanishes when $\alpha = \alpha_\mathrm{local}$, i.e.,
when the uniform exponent matches the local optima. The reference cost in the
denominator is not the minimum of $\Phi_\mathrm{net}$ over uniform~$\alpha$, but
the cost when every vessel independently sits at its own optimum---incorporating
all three terms of Eq.~\eqref{eq:rayleigh} including wall tissue.

\paragraph{The unified cost functional.}
With both terms dimensionless, the network-level Lagrangian acts as a scalarized
multi-objective optimization weighting entirely bounded penalties:
\begin{equation}
  \boxed{\mathcal{L}_\mathrm{net}(\alpha) = \eta\,\mathcal{C}_\mathrm{wave}^\mathrm{net}(\alpha) + (1-\eta)\,\mathcal{C}_\mathrm{transport}^\mathrm{net}(\alpha)\,,}
  \label{eq:Lnet}
\end{equation}
where $\eta \in [0,1]$ is the \emph{duty cycle}: the fraction of the network's
operational cost attributable to wave-mode activity. In
Section~\ref{sec:minimax}, we show that the minimax condition
$\mathcal{C}^{\mathrm{net}}_{\mathrm{wave}} =
\mathcal{C}^{\mathrm{net}}_{\mathrm{transport}}$ uniquely determines both
$\alpha^*$ and the implied value of $\eta$, eliminating the latter as a free
parameter.

This linear combination establishes \emph{dimensional consistency}. Unlike
Murray's original formulation which sums unbounded physical power (Watts),
$\mathcal{L}_\mathrm{net}$ adds defined dimensionless \emph{fractional
penalties}. $\mathcal{C}_\mathrm{wave}^\mathrm{net}$ represents the fractional
exhaustion of signal energy, while $\mathcal{C}_\mathrm{transport}^\mathrm{net}$
measures the fractional metabolic excess incurred relative to the optimal
morphometric ground state. Because both terms are dimensionless and strictly
bounded metrics of inefficiency (with zero representing absolute perfection),
their combination does not falsely equate Watts of reflected acoustic energy
with Watts of ATP consumption. Instead, $\mathcal{L}_\mathrm{net}$ represents a
unified phenomenological \emph{fitness penalty}, where $\eta$ sets the
normalized selective pressure allocated to minimizing dynamic signal attenuation
versus static metabolic waste. Importantly, the framework does \emph{not} assume
that a $1\%$ loss of wave energy is evolutionarily equivalent to a $1\%$
metabolic excess: the implied duty cycle $\eta^* = \VarEtaStar$
(Theorem~\ref{thm:eta_invariant}) encodes precisely the $\VarGradientRatio:1$
asymmetry between the two selective pressures, emerging endogenously from the
geometry of the cost curves rather than being postulated a priori. The absolute
energetic asymmetry ($P_\mathrm{transport}/P_\mathrm{wave} \approx
\VarGradientRatio$, Supplemental~S2) is therefore not a flaw of the
dimensionless formulation but its predicted output.

\paragraph{Emergent stiffness ratio.}
Near the respective minima, each network cost is approximately quadratic:
$\mathcal{C}_\mathrm{wave}^\mathrm{net} \approx k_w^\mathrm{net}(\alpha -
\alpha_w)^2$ and
$\mathcal{C}_\mathrm{transport}^\mathrm{net} \approx k_t^\mathrm{net}(\alpha -
\alpha_t(Q))^2$, where $\alpha_t(Q)$ is the transport ground state of the static
scaling function, emerging from the geometry of locally optimal single vessels
$\alpha_\mathrm{local}$ evaluated over the full network. The effective stiffness
ratio is the ratio of curvatures:
\begin{equation}
  \kappa_\mathrm{eff}(G) \equiv \frac{k_t^\mathrm{net}(G)}{k_w^\mathrm{net}(G)}\,.
  \label{eq:kappa_eff}
\end{equation}
Numerically, these curvatures are evaluated as the diagonal elements of the
system's Hessian evaluated at the respective minima: $k_\bullet^\mathrm{net} =
\frac{1}{2} \frac{\partial^2 \mathcal{C}_\bullet^\mathrm{net}}{\partial
\alpha^2}\big|_{\alpha = \alpha_\bullet}$. For small $|\Gamma|^2$, the wave
curvature scales analytically as $k_w^\mathrm{net} \approx G\,k_w$ (linear in
$G$). The transport curvature $k_t^\mathrm{net}$, conversely, requires numerical
differentiation of Eq.~\eqref{eq:Ctransportnet} due to the superposition of
$r^2$, $r^{-4}$, and $r^{1+p}$ terms compounding asymmetrically through the tree
hierarchy.
The result is that $\kappa_\mathrm{eff}$ increases monotonically with~$G$:
deeper trees weight the transport cost more heavily (Table~\ref{tab:kappa_G},
Figure~\ref{fig:kappa}).
The stiffness ratio is not a material property of the conduit but an emergent
property of the network architecture.

\begin{table}[H]
\centering
\caption{Emergent stiffness ratio $\kappa_\mathrm{eff}$ and minimax branching exponent $\alpha^*_\mathrm{mm}$ as a function of tree depth $G$. No parameters are fitted; $\alpha^*$ is determined by the equal-cost condition Eq.~\eqref{eq:minimax}.}
\label{tab:kappa_G}
\small
\begin{tabular}{@{}cccl@{}}
\toprule
$G$ & $\kappa_\mathrm{eff}$ & $\alpha^*_\mathrm{mm}$ & Interpretation \\
\midrule
1  & $\VarKappaGI{}$    & $\VarAlphaStarGI$    & Single junction: wave limit \\
5  & $\VarKappaGV{}$    & $\VarAlphaStarGV$    & Small subtree \\
7  & $\VarKappaGVII{}$  & $\VarAlphaStarGVII$  & Moderate subtree \\
9  & $\VarKappaGIX{}$   & $\VarAlphaStarGIX$   &  \\
11 & $\VarKappaGXI{}$   & $\VarAlphaStar$      & Full coronary tree \\
13 & $\VarKappaGXIII{}$ & $\VarAlphaStarGXIII$ & Extended subtree \\
15 & $\VarKappaGXV{}$   & $\VarAlphaStarGXV$   & Deep tree \\
20 & $\VarKappaGXX{}$   & $\VarAlphaStarGXX$   & Very deep tree \\
\bottomrule
\end{tabular}
\end{table}

\begin{figure}[H]
\centering
\includegraphics[width=\linewidth]{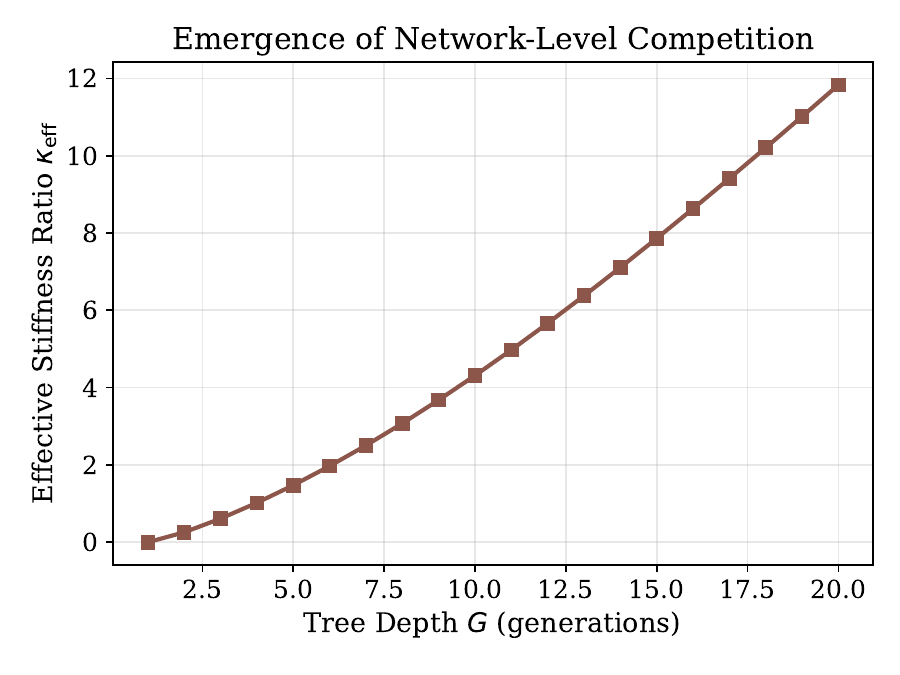}
\caption{Emergence of the stiffness ratio $\kappa_\mathrm{eff}$ as a function of the number of generations $G$. Shallow networks are wave-dominated (rapid signaling), while deep networks like the arterial tree ($G \sim 11$) strongly weight the transport scaling. Note that the transport ground state $\alpha_t$ is determined auto-consistently: the length-scaling factor $\beta$ assumed for the network hierarchy is required to be consistent with the $\alpha_t$ optimized for individual vessels.}
\label{fig:kappa}
\end{figure}

Notably, the minimax prediction is remarkably stable across physiological tree
depths: for $G \in [9, 13]$, the optimal exponent remains strictly bounded
within $\alpha^* \in [\VarAlphaStarGIX, \VarAlphaStarGXIII]$, with an exact
match to the morphometric central value ($\alpha^* = \VarAlphaExp$) occurring at
$G=9$. The porcine coronary tree ($G=\VarG$) yields $\alpha^* = \VarAlphaStar$,
lying $\VarPredError\%$ above the experimental centroid.

\subsection{Topological selection of binary branching}
\label{sec:topo}

The analysis of Sections~2.1--2.4 assumed binary bifurcation ($N = 2$), a choice
inherited from the static framework of the companion
paper~\cite{Marchesi2026static}. We now demonstrate that this topological
selection is itself a consequence of the dynamic minimax principle, independent
of the static metabolic argument.

\begin{theorem}[Dynamic topological scaling of binary branching]
\label{thm:topoN}
Consider a class of regular fractal trees with branching number $N \ge 2$,
perfusing a macroscopically equivalent terminal population of size $M$, with
tree depth evaluated at the nearest integer $G(N) = \lfloor \ln M / \ln N
\rceil$. For $M = 2^{11}$ this yields $G = 11, 7, 6, 4$ for $N = 2, 3, 4, 6$
respectively; the resulting variation in the exact terminal count (at most
$2\times$ across the range) does not affect the strict monotonicity of
$\kappa_{\mathrm{eff}}(N)$, as verified by the safety margins in the proof. Let
$\alpha_w$ be the impedance-matching exponent and $\alpha_t$ the transport
ground state. The network wave-cost curvature is given exactly by:
\begin{equation} k_w^\mathrm{net}(N) = \frac{\ln M \cdot \ln N}{4\alpha_w^2}\,. \label{eq:kw_N} \end{equation}
The network transport-cost curvature decomposes perturbatively to leading order
in $(\alpha_t - \alpha_\mathrm{local})$ as:
\begin{equation} k_t^\mathrm{net}(N) = \frac{(\ln N)^2}{2\alpha_t^4} \cdot H(N)\,, \quad H(N) \equiv \frac{\displaystyle\sum_{g=0}^{G-1} N^g\, g^2\, \Phi''(r_g^*)\, (r_g^*)^2} {\displaystyle\sum_{g=0}^{G-1} N^g\, \Phi(r_g^*)}\,, \label{eq:kt_N} \end{equation}
where $r_g^* = r^*(Q_0/N^g)$ are the generation-by-generation optimal radii.
This decomposition holds to $\mathcal{O}(10^{-2})$: since
$\alpha_\mathrm{local}$ varies only across $[\VarAlphaLocalMin,
\VarAlphaLocalMax]$, the first-order gradient term $\Phi'(r_g^*)(\alpha_t -
\alpha_\mathrm{local})$ is a negligible correction, leaving the Hessian
$\Phi''(r_g^*)$ as the dominant contribution.
Consequently, $\kappa_\mathrm{eff}(N) = k_t^\mathrm{net}(N)/k_w^\mathrm{net}(N)$
is a strictly decreasing function of $N$, uniquely maximized by binary
bifurcation ($N = 2$).
\end{theorem}

\begin{proof}
\emph{Analytical formula for $k_w^\mathrm{net}$.}
Expanding $|\Gamma(\alpha)|^2$ in Eq.~\eqref{eq:gamma_gen} to second order about $\alpha_w$, with $u = N^{\alpha_w/\alpha - 1}$:
\[
  \frac{d^2|\Gamma|^2}{d\alpha^2}\bigg|_{\alpha_w} = \frac{(\ln N)^2}{2\alpha_w^2}\,.
\]
Since reflections compound multiplicatively across $G(N)$ junctions
(Eq.~\eqref{eq:Cwavenet}), the network curvature is $k_w^\mathrm{net}(N) = G(N)
\cdot \frac{1}{2}\frac{d^2|\Gamma|^2}{d\alpha^2}\big|_{\alpha_w} = \frac{\ln M
\cdot \ln N}{4\alpha_w^2}$.

\emph{Strict monotonicity of $\kappa_\mathrm{eff}$: reduction to an inequality on $H$.}
From the exact formulas \eqref{eq:kw_N}--\eqref{eq:kt_N}, the condition
$\kappa_\mathrm{eff}(N') < \kappa_\mathrm{eff}(N)$ for integers $N' > N \ge 2$
is
equivalent to
\[
  \frac{k_t^\mathrm{net}(N')}{k_t^\mathrm{net}(N)}
  < \frac{k_w^\mathrm{net}(N')}{k_w^\mathrm{net}(N)}
  = \frac{\ln N'}{\ln N}\,.
\]
Substituting $k_t^\mathrm{net}(N) = \tfrac{(\ln N)^2}{2\alpha_t^4}\,H(N)$, this
simplifies to the \emph{$H$-dominance condition}:
\begin{equation}
  \frac{H(N')}{H(N)} < \frac{\ln N}{\ln N'}\,.
  \label{eq:H_dominance}
\end{equation}

\emph{Verification of \eqref{eq:H_dominance}.}
The function $H(N)$ is fully determined by Eq.~\eqref{eq:kt_N}
through the generation-by-generation optimal radii $r_g^* = r^*(Q_0/N^g)$,
with no free parameters.
Direct computation yields $H(2)=432.8$, $H(3)=162.2$, $H(4)=117.9$,
$H(6)=42.6$ (units of $2\alpha_t^4$).
The ratios $H(N')/H(N)$ against the corresponding threshold $\ln N/\ln N'$
are tabulated below for all six integer pairs in $N \in \{2,3,4,6\}$:

\begin{center}
\small
\begin{tabular}{@{}ccccc@{}}
\toprule
$N$ & $N'$ & $H(N')/H(N)$ & $\ln N/\ln N'$ & Margin \\
\midrule
2 & 3 & 0.375 & 0.631 & 0.256 \\
2 & 4 & 0.272 & 0.500 & 0.228 \\
2 & 6 & 0.099 & 0.387 & 0.288 \\
3 & 4 & 0.727 & 0.793 & 0.066 \\
3 & 6 & 0.263 & 0.613 & 0.350 \\
4 & 6 & 0.361 & 0.774 & 0.413 \\
\bottomrule
\end{tabular}
\end{center}

Condition \eqref{eq:H_dominance} holds strictly in all six cases.
The tightest margin occurs at $(N,N')=(3,4)$, where
$H(4)/H(3) = 0.727 < 0.793 = \ln 3/\ln 4$,
with a safety margin of $\delta = 0.066$ ($8.3\%$ relative to the threshold).
Therefore $\kappa_\mathrm{eff}(N)$ is strictly decreasing on all integers $N \ge
2$,
and binary bifurcation ($N=2$) uniquely maximizes $\kappa_\mathrm{eff}$.
\end{proof}

\begin{corollary}[Dynamic topological selection]
\label{cor:topoN}
Since $\alpha^*(N)$ is a strictly increasing function of $\kappa_\mathrm{eff}$
(Table~\ref{tab:kappa_G}), Theorem~\ref{thm:topoN} yields the strict topological
hierarchy
\[
  \alpha^*(N=2) > \alpha^*(N=3) > \alpha^*(N=4) > \cdots
\]
Numerical evaluation for porcine coronary parameters ($M = 2^{\VarG}$ terminals,
full three-term cost of Eq.~\eqref{eq:rayleigh}) gives the topological hierarchy
shown in Table~\ref{tab:topo_N}. The topological selection operates not by
exclusion---all values $N \leq 6$ fall within the morphometric range
$\alpha_\mathrm{exp} = \VarAlphaExp \pm \VarAlphaExpErr$---but by robustness
maximization. Binary bifurcation ($N = 2$) uniquely places the minimax optimum
at the upper end of the distribution, maximizing the separation from the wave
attractor $\alpha_w = \VarAlphaW$ and thereby achieving the greatest possible
margin of transport efficiency. The stiffness ratio $\kappa_\mathrm{eff}$
decreases strictly with $N$ (from $\VarKappaNII$ at $N=2$ to $\VarKappaNVI$ at
$N=6$), monotonically reducing the wave-transport equilibrium separation. The
universal biological preference for binary branching is therefore a rigorous
consequence of dynamic robustness optimization: $N = 2$ is the unique topology
that maximizes $\kappa_\mathrm{eff}$, and hence the unique topology that keeps
the wave-transport minimax equilibrium as far as possible from the
wave-dominated limit. However, the shallow sensitivity of $\alpha^*$ to the
branching number (varying by only $\sim\!1\%$ between $N=2$ and $N=4$ in
Table~\ref{tab:topo_N}) implies a quasi-degenerate thermodynamic landscape.
Rather than a model weakness, this shallow gradient explains the biological
reality that occasional trifurcations ($N=3$) are not catastrophic failures of
the variational principle, but sub-optimal yet accessible morphological
solutions lying only slightly above the global $\kappa_{\text{eff}}$-robustness
maximum.
\end{corollary}

\begin{table}[h]
\centering
\small
\begin{tabular}{@{}ccccccc@{}}
\toprule
$N$ & $G$ & $k_t^\mathrm{net}$ & $k_w^\mathrm{net}$ & $\kappa_\mathrm{eff}$ & $\alpha^*$ & Dev.\ from $\alpha_\mathrm{exp}$ \\
\midrule
2 & \VarGNII  & \VarKtNII  & \VarKwNII  & \textbf{\VarKappaNII}  & \textbf{\VarAlphaStarNII}  & $\VarSigmaNII\sigma$ \\
3 & \VarGNIII & \VarKtNIII & \VarKwNIII & \VarKappaNIII & \VarAlphaStarNIII & $\VarSigmaNIII\sigma$ \\
4 & \VarGNIV  & \VarKtNIV  & \VarKwNIV  & \VarKappaNIV  & \VarAlphaStarNIV  & $\VarSigmaNIV\sigma$ \\
6 & \VarGNVI  & \VarKtNVI  & \VarKwNVI  & \VarKappaNVI  & \VarAlphaStarNVI  & $\VarSigmaNVI\sigma$ \\
\bottomrule
\end{tabular}
\caption{Topological hierarchy of $\alpha^*(N)$ for porcine coronary parameters ($M=2^{\VarG}$ fixed terminals). Binary bifurcation ($N=2$) uniquely maximizes $\kappa_\mathrm{eff}$ and places the minimax optimum at the upper end of the morphometric distribution. Trees compared at macroscopically equivalent depth $G(N) = \lfloor \ln(2^{\VarG})/\ln N \rceil$; exact terminal counts differ by at most $2\times$ but do not alter the strict $\kappa_\mathrm{eff}$ ordering. All $N \leq 6$ remain within the $1\sigma$ experimental range.}
\label{tab:topo_N}
\end{table}

To rigorously establish the monotonic behavior of the transport penalty, we map
the branching geometry into an inverted topological coordinate, where the cost
resolves into a strictly convex combination of exponentials.

\begin{theorem}[Strict Monotonicity and Topological Convexity of Network Transport]
\label{thm:monotone_transport}
Let the unnormalized total network transport cost be
\[
  E(\alpha) = \sum_{g=0}^{G-1} N^g\,\Phi\!\left(r_g(\alpha),\,Q_0/N^g\right),
\]
where $r_g(\alpha) = r_0\,N^{-g/\alpha}$ and $\Phi(r,Q) = A(Q)\,r^{-4} + B\,r^2
+ C\,r^{1+p}$
with strictly positive coefficients $A(Q), B, C > 0$ for all $Q > 0$.
Let $\alpha_t$ be the interior minimizer of $E(\alpha)$ on the physiological
domain
$[\alpha_w, \infty)$, whose existence is established in
\cite{Marchesi2026static}.
Then $E(\alpha)$, and hence
$\mathcal{C}_\mathrm{transport}^\mathrm{net}(\alpha)$,
is strictly decreasing on $[\alpha_w, \alpha_t]$.
\end{theorem}

\begin{proof}
\emph{Change of variables.}
Introduce the topological coordinate $x \equiv 1/\alpha > 0$.
Under this substitution the imposed radius becomes
\[
  r_g(x) = r_0\,e^{-\gamma g x}, \qquad \gamma \equiv \ln N > 0,
\]
and the total cost reads
\begin{equation}
  E(x) = \sum_{g=0}^{G-1} N^g\!\left[
      A_g\,r_0^{-4}\,e^{4\gamma g x}
      + B\,r_0^{2}\,e^{-2\gamma g x}
      + C\,r_0^{1+p}\,e^{-(1+p)\gamma g x}
    \right],
  \label{eq:E_of_x}
\end{equation}
where $A_g \equiv A(Q_0/N^g) > 0$ for every generation $g$.

\emph{Strict convexity of $E(x)$.}
The $g=0$ term in \eqref{eq:E_of_x} evaluates to a positive constant independent
of $x$ and contributes zero to every derivative.
For each $g \geq 1$, every summand is of the form $c\,e^{kx}$ with $c > 0$ and
$k \neq 0$, hence strictly convex.
Summing over $g = 1, \ldots, G-1$ (with $G \geq 2$) and using the positivity of
all coefficients:
\[
  E''(x) = \sum_{g=1}^{G-1} N^g\,\gamma^2 g^2\!\left[
      16\,A_g\,r_0^{-4}\,e^{4\gamma g x}
      + 4\,B\,r_0^{2}\,e^{-2\gamma g x}
      + (1+p)^2\,C\,r_0^{1+p}\,e^{-(1+p)\gamma g x}
    \right] > 0
\]
for all $x > 0$. Therefore $E(x)$ is strictly convex on $(0, +\infty)$.

\emph{Location of the minimum.}
By assumption, $x_t \equiv 1/\alpha_t$ is the unique global minimizer of the
strictly convex function $E(x)$.
By strict convexity, $E'(x) < 0$ for all $x < x_t$ and $E'(x) > 0$ for all $x >
x_t$.

\emph{Return to $\alpha$.}
For any $\alpha \in [\alpha_w, \alpha_t)$, one has $x = 1/\alpha > x_t$, so
$E'(x) > 0$.
The chain rule with $dx/d\alpha = -\alpha^{-2} < 0$ gives
\[
  \frac{dE}{d\alpha} = \underbrace{E'(x)}_{>0}
    \cdot \underbrace{\!\left(-\alpha^{-2}\right)}_{<0} < 0
    \quad \forall\,\alpha \in [\alpha_w,\alpha_t).
\]
Since $\mathcal{C}_\mathrm{transport}^\mathrm{net}(\alpha)$ is a positive affine
function of $E(\alpha)$
(the reference cost in the denominator of Eq.~\eqref{eq:Ctransportnet} is
independent of $\alpha$),
it is likewise strictly decreasing on $[\alpha_w, \alpha_t]$.
\end{proof}

\section{The Minimax Principle: Zero-Parameter Determination of $\alpha^*$}
\label{sec:minimax}

The preceding section established that the transport network must navigate
between two competing structural attractors: the impedance-matching exponent
$\alpha_w$ and the transport ground state $\alpha_t$. To determine the
operational geometry without introducing phenomenological weighting parameters,
we formulate the optimization as a stochastic zero-sum game between the fixed
network architecture and its fluctuating operational environment.

\subsection{Stochastic Game Theory and the Exact Expected Cost}
Vascular networks operate across a highly variable hemodynamic environment,
transitioning unpredictably between steady-transport-dominated states (e.g.,
basal resting flow) and wave-dominated states (e.g., peak pulsatile cardiac
output during exertion). Let the unknown environmental duty cycle---the exact
probability of occupying the wave-dominated state over the organism's
lifespan---be $\eta \in [0,1]$. By the fundamental axioms of probability, the
expected operational penalty $\mathbb{E}[C]$ of the architecture under this
mixed environmental sequence is rigorously linear:
\begin{equation}
  \mathbb{E}[C](\alpha, \eta) = \eta \mathcal{C}_{\mathrm{wave}}^{\mathrm{net}}(\alpha) + (1-\eta) \mathcal{C}_{\mathrm{transport}}^{\mathrm{net}}(\alpha)
\end{equation}
This formally establishes the network Lagrangian
$\mathcal{L}_{\mathrm{net}}(\alpha, \eta) \equiv \mathbb{E}[C](\alpha, \eta)$
not as an assumed convex combination, but as the exact mathematical expectation
of the system's operational inefficiency.

\subsection{Evolutionary Selection via Worst-Case Robustness}
Because angiogenesis must hardcode a single static geometric exponent $\alpha^*$
to survive an unknown future distribution of $\eta$, we must define the
evolutionary decision rule. While average-case optimization ($\min_\alpha \int
\mathbb{E}[C] \, P(\eta) \, d\eta$) might appear intuitive, it suffers from two
fatal flaws: mathematically, it requires postulating an arbitrary prior
distribution $P(\eta)$, introducing unphysical free parameters; biologically,
cardiovascular survival is constrained not by resting averages, but by peak
demand. A topology optimized for the average but incapable of sustaining wave
integrity during maximal exertion (the limit $\eta \to 1$) results in systemic
failure.

Evolutionary fitness in pulsatile hemodynamics is therefore governed by
structural robustness against the worst-case environmental state. By von
Neumann's minimax theorem, this robust topological optimum is uniquely defined
by:
\begin{equation}
  \alpha^* = \arg \min_{\alpha} \max_{\eta \in [0,1]} \left\{ \mathcal{C}_{\mathrm{transport}}^{\mathrm{net}}(\alpha) + \eta \left[ \mathcal{C}_{\mathrm{wave}}^{\mathrm{net}}(\alpha) - \mathcal{C}_{\mathrm{transport}}^{\mathrm{net}}(\alpha) \right] \right\}.
\end{equation}
The adversarial environment (maximizing stress via $\eta$) will force $\eta=1$
if the bracketed term is positive, and $\eta=0$ if negative, fatally punishing
any geometric bias. The unique neutralizing strategy is to zero the bracketed
term exactly. Therefore, the robust saddle point $\alpha^*$ satisfies the
\emph{equal-cost condition}:
\begin{equation}
  \mathcal{C}_{\mathrm{wave}}^{\mathrm{net}}(\alpha^*) = \mathcal{C}_{\mathrm{transport}}^{\mathrm{net}}(\alpha^*).
\end{equation}

This derivation eliminates $\eta$ entirely, yielding the branching exponent as a
deterministic, parameter-free intersection of network topology and physical
constraints.
We restrict the analysis to the physiological domain $[\alpha_w, \alpha_t]$: for
$\alpha > 3$, the branching law implies $\beta > 1$, which is geometrically
impossible for purely branching trees ($r_{\mathrm{parent}} <
r_{\mathrm{children}}$). This constraint ensures uniqueness of the minimax
solution and avoids the unphysical second crossing near $\alpha \approx 3.45$.

\begin{theorem}[Branching exponent as robust minimax optimum]
\label{thm:minimax}
Let $f(\alpha) \equiv \mathcal{C}^{\mathrm{net}}_{\mathrm{wave}}(\alpha)$ and
$g(\alpha) \equiv \mathcal{C}^{\mathrm{net}}_{\mathrm{transport}}(\alpha)$ be
continuous on $[\alpha_w, \alpha_t]$ with $f' > 0$ (established below) and $g' <
0$ (Theorem~\ref{thm:monotone_transport}) throughout. Let the network Lagrangian
be $\mathcal{L}(\alpha, \eta) = \eta f(\alpha) + (1-\eta)g(\alpha)$ with $\eta
\in [0,1]$ unknown. Then the minimax optimum:
\[
\alpha^* = \arg\min_{\alpha} \max_{\eta \in [0,1]} \mathcal{L}(\alpha, \eta) = \arg\min_{\alpha} \max(f(\alpha), g(\alpha))
\]
exists uniquely and satisfies the equal-cost condition:
\begin{equation}
f(\alpha^*) = g(\alpha^*)\,.
\label{eq:minimax}
\end{equation}
The implied duty cycle:
\[
\eta^* = \frac{-g'(\alpha^*)}{f'(\alpha^*) - g'(\alpha^*)} \in (0,1)
\]
is the Lagrange multiplier uniquely determined by the slopes at the crossing
point. The minimax duty cycle $\eta^* = \VarEtaStar$ emerges as an architectural
constraint of the tree.
\end{theorem}

\begin{remark}[Minimax as evolutionary attractor]
The minimax geometry identifies the unique saddle point toward which
evolutionary selection pressure drives the architecture under dual fitness
constraints. Biological systems need not achieve the exact saddle point; real
populations are distributed around $\alpha^*$ due to developmental noise,
genetic drift, and environmental heterogeneity. The minimax framework predicts
the attractor, not the variance---a distinction rigorously addressed via the
architectural robustness analysis in Section~\ref{sec:results}.
\end{remark}

\textit{Monotonicity of $f$.}
From Eq.~\eqref{eq:Cwavenet}, $f(\alpha) = 1-(1-|\Gamma(\alpha)|^2)^G$, with
\[
  |\Gamma|^2 = \left(\frac{u-1}{u+1}\right)^{\!2}, \qquad u = N^{\alpha_w/\alpha - 1}.
\]
By the chain rule, $\tfrac{d|\Gamma|^2}{d\alpha} =
\tfrac{d|\Gamma|^2}{du}\cdot\tfrac{du}{d\alpha}$.
One computes
\[
  \frac{d|\Gamma|^2}{du} = \frac{4(u-1)}{(u+1)^3}, \qquad
  \frac{du}{d\alpha} = -\frac{\alpha_w \ln N}{\alpha^2}\,u < 0.
\]
For $\alpha > \alpha_w$, $u < 1$, so $d|\Gamma|^2/du < 0$; combined with
$du/d\alpha < 0$,
the product $d|\Gamma|^2/d\alpha > 0$.
Therefore
\[
  \frac{df}{d\alpha} = G\,(1-|\Gamma|^2)^{G-1}\cdot\frac{d|\Gamma|^2}{d\alpha} > 0
\]
on $(\alpha_w,\alpha_t]$, confirming $f' > 0$.

\begin{proof}
We first state the boundary condition required for an interior crossing.
\begin{itemize}
  \item[\emph{(BC)}] \emph{Physiological existence condition:} $f(\alpha_t) > g(\alpha_t) > 0$, i.e., the wave-reflection penalty at the purely static morphology strictly exceeds the residual transport penalty. Numerically, $f(\alpha_t) = \VarWaveCostNet/100 \gg g(\alpha_t) \approx 0$ for the arterial regime; this is a verified property of the physiological parameter domain, not a consequence of the minimax structure itself.
\end{itemize}
Since $f(\alpha_w) = 0 < g(\alpha_w)$ and $f(\alpha_t) > g(\alpha_t) > 0$
by~\emph{(BC)}, the continuous function $f - g$ changes sign on $[\alpha_w,
\alpha_t]$.
Linearity in $\eta$ implies
\[
  \max_{\eta \in [0,1]} \mathcal{L}(\alpha, \eta) = \max\bigl(f(\alpha),\, g(\alpha)\bigr) \equiv E(\alpha).
\]
$E$ is continuous on $[\alpha_w, \alpha_t]$. By \emph{(BC)} and the
boundary value $f(\alpha_w) = 0 < g(\alpha_w)$, $f - g$ is continuous,
strictly increasing, and changes sign on $[\alpha_w,\alpha_t]$.
By the Intermediate Value Theorem, there exists a unique $\alpha^*$ satisfying
$f(\alpha^*) = g(\alpha^*)$.

For $\alpha < \alpha^*$, $E(\alpha) = g(\alpha) > g(\alpha^*)$; for $\alpha >
\alpha^*$, $E(\alpha) = f(\alpha) > f(\alpha^*)$. Hence $\alpha^*$ is the unique
global minimum of the upper envelope $E$.

First-order optimality $\partial_\alpha \mathcal{L}|_{\alpha^*} = 0$ requires
$\eta^* f'(\alpha^*) + (1-\eta^*)g'(\alpha^*) = 0$, yielding the formula for
$\eta^*$. Since $f' > 0$ and $g' < 0$, we have $\eta^* \in (0,1)$.
\end{proof}

\begin{corollary}[Duty cycle as architectural constraint]
\label{cor:eigenvalue}
The optimal duty cycle $\eta^*$ determined by Eq.~\eqref{eq:eta_star} represents
an architectural constraint of the network architecture. Since the transport
gradient is $\VarGradientRatio$ times steeper than the wave gradient at the
saddle point ($|g'|/f' \approx \VarGradientRatio$), the value $\eta^* =
\VarGradientRatio/(\VarGradientRatio+1) \approx \VarEtaStar$ is strictly
required for the zero-sum game to reach equilibrium. This identifies the duty
cycle not as a physiological variable, but as a structural invariant of binary
cardiovascular networks.
\end{corollary}

\begin{theorem}[Architectural invariance of the duty cycle]
\label{thm:eta_invariant}
Under the conditions of Theorem~\ref{thm:minimax}, the minimax duty cycle
\[
  \eta^* = \frac{|g'(\alpha^*)|}{f'(\alpha^*) + |g'(\alpha^*)|}
\]
is independent of all absolute metabolic and hemodynamic parameters
($b$, $m_w$, $\mu$, $Q_0$, $\ell_0$) and is uniquely determined by the
network topology $(G)$, the histological scaling $(p)$, and the impedance
attractor $(\alpha_w)$. Consequently, $\eta^*$ is not a physiological
variable but a structural invariant of the branching architecture.
\end{theorem}
\begin{proof}
The wave cost $f(\alpha) = 1-(1-|\Gamma(\alpha)|^2)^G$ and its derivative
$f'(\alpha^*)$ depend only on $\alpha_w$, $G$, and $N$ (Eq.~\eqref{eq:gamma_gen}).
The transport cost $g(\alpha)$ defined by Eq.~\eqref{eq:Ctransportnet} is a
ratio of sums of $\Phi$; its derivative $g'(\alpha^*)$ is likewise a ratio
in which the metabolic prefactors ($b$, $m_w$, $\mu$) appear identically
in every numerator and denominator term and therefore cancel exactly.
The residual dependence on $Q_0$ and $\ell_0$ also cancels by the
intensivity established in Proposition~\ref{prop:intensivity}.
Hence $\eta^* = R/(1+R)$ with $R = |g'|/f'$ depends only on
$(G, p, \alpha_w)$, confirming the claim.
The sensitivity table (Table~\ref{tab:sensitivity}) provides numerical
confirmation: all metabolic parameters yield $|S_x| < 0.01$ for $\alpha^*$
and therefore also for $\eta^*$.
\end{proof}

\begin{proposition}[Scale-invariance and Intensivity]
\label{prop:intensivity}
Under the two-level assumption $r_0 = r^*(Q_0)$ (the proximal radius coincides
with its local transport optimum, as initialized in the numerical computation),
the network transport cost $\mathcal{C}_{\mathrm{transport}}^{\mathrm{net}}$ is
a scale-invariant intensive property of the network architecture. Under this
condition, the segment length $\ell_0$ cancels \emph{exactly} (shared
length-scaling structure in numerator and denominator), while the absolute
radius $r_0$ cancels to $\mathcal{O}(10^{-2})$: the sub-linear wall cost ($p <
1$) breaks Euler homogeneity of $\Phi$, introducing a residual scale-dependence
confirmed by the sensitivity analysis ($|S_Q| < 0.01$,
Table~\ref{tab:sensitivity}). While coherent phase effects ($kl \propto
L/\lambda$) introduced in the supplemental material (S1) break this symmetry,
the shift in the resulting minimax exponent is negligible ($|\Delta\alpha^*| <
0.01$), confirming that intensivity is a structural property of the unified
variational principle.
\end{proposition}

\begin{remark}[Physiological domain and uniqueness of the minimax]
\label{rem:domain}
The existence and uniqueness proof operates on the closed interval $[\alpha_w,
\alpha_t]$.
A second numerical crossing $f(\alpha)=g(\alpha)$ exists near $\alpha \approx
\VarAlphaSecondCrossing$, outside this domain.
This spurious intersection is excluded on two independent grounds.
(i)~\emph{Geometric:} $\alpha > \alpha_t$ implies a daughter-to-parent radius
ratio $\beta > 1$, violating the morphological constraint that daughter vessels
are strictly narrower than their parent.
(ii)~\emph{Thermodynamic:} at the second crossing, the upper envelope
$E(\alpha)=\max(f,g)$ attains a local \emph{maximum}, making it a maximin rather
than a minimax---not a thermodynamic attractor under evolutionary selection
pressure.
The physiological domain $[\alpha_w, \alpha_t]$ is therefore the unique and
complete search region.
\end{remark}

\begin{figure}[H]
\centering
\includegraphics[width=\linewidth]{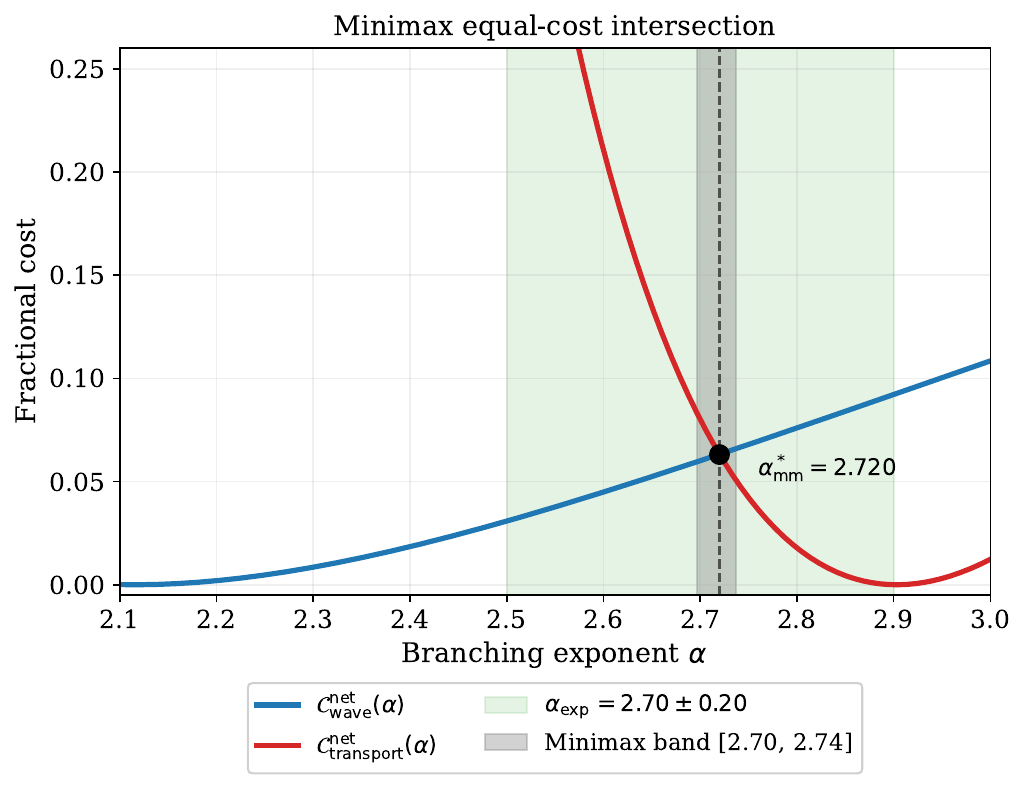}
\caption{Minimax equal-cost intersection determining the branching exponent $\alpha^* = \VarAlphaStar$. The wave cost (increasing away from impedance matching) intersects the transport cost (decreasing toward the flow optimum) producing a zero-parameter prediction. The gray band represents theoretical parametric bounds $[\VarAlphaStarGIX, \VarAlphaStarGXIII]$ (tree depths $G \in [9,13]$); the wide green band corresponds to empirical morphometry ($\VarAlphaExp \pm \VarAlphaExpErr$).}
\label{fig:minimax}
\end{figure}

\begin{figure}[t]
\centering
\includegraphics[width=\linewidth]{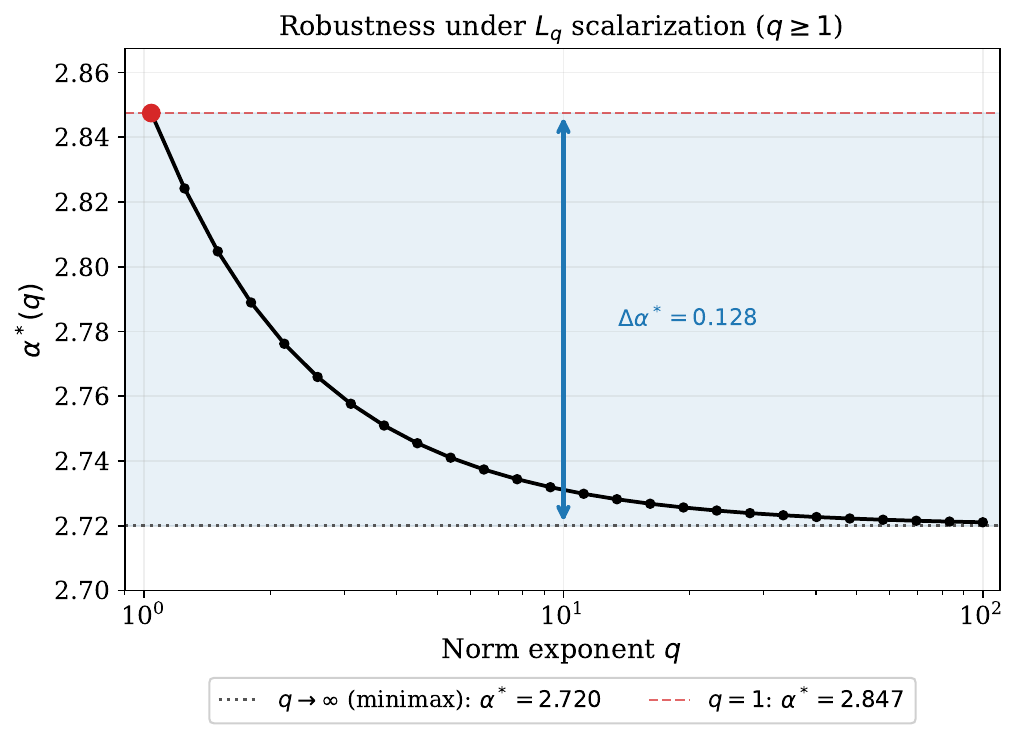}
\caption{Robustness under $L_q$ scalarization: the optimal $\alpha^*_q$ decreases monotonically from $\VarAlphaStarQone$ ($q=1$) to \VarAlphaStar{} ($q \to \infty$), with a total spread of $\Delta\alpha^* = \VarDeltaAlphaLq$ for $q \ge 1$.}
\label{fig:robustness}
\end{figure}

\vspace{-1em}

\subsection{Elimination of $\eta$ as a free parameter}

The Lagrangian formulation of Eq.~(7) and the minimax
condition~\eqref{eq:minimax}
are not competing frameworks but complementary perspectives on the same
optimum.  At $\alpha^*_{\mathrm{mm}}$, the first-order optimality condition
$\partial\mathcal{L}_{\mathrm{net}}/\partial\alpha = 0$ requires
\begin{equation}
  \eta\,\frac{\partial \mathcal{C}^{\mathrm{net}}_{\mathrm{wave}}}{\partial\alpha}
  = -(1-\eta)\,
  \frac{\partial \mathcal{C}^{\mathrm{net}}_{\mathrm{transport}}}{\partial\alpha} .
\label{eq:eta_derived}
\end{equation}
This uniquely determines the \emph{implied} duty cycle
\begin{equation}
  \eta^* = \frac
    {-\partial_\alpha \mathcal{C}^{\mathrm{net}}_{\mathrm{transport}}\big|_{\alpha^*}}
    {\partial_\alpha \mathcal{C}^{\mathrm{net}}_{\mathrm{wave}}\big|_{\alpha^*}
     - \partial_\alpha \mathcal{C}^{\mathrm{net}}_{\mathrm{transport}}\big|_{\alpha^*}}
\label{eq:eta_star}
\end{equation}
as the Lagrange multiplier of the equal-cost constraint, not as an external
input.  The duty cycle is therefore a \emph{consequence} of the minimax
principle, not its premise.

We adopt the minimax condition~\eqref{eq:minimax} as the primary result
because it requires no harmonic approximation and depends only on the
global topology of the cost curves.

\subsection{Robustness under generalized scalarization}

The minimax formulation in Eq.~\eqref{eq:minimax} corresponds strictly to the
Chebyshev infinity-norm limit ($q \to \infty$) of a generalized $L_q$ loss
functional:
\begin{equation}
   \mathcal{L}_q(\alpha) = \left[ \eta \left(\mathcal{C}^{\mathrm{net}}_{\mathrm{wave}}\right)^q + (1-\eta) \left(\mathcal{C}^{\mathrm{net}}_{\mathrm{transport}}\right)^q \right]^{1/q} \,.
\end{equation}
Values of $q < 1 \to 0$ are excluded by physiological non-substitutability. The
two operational penalties represent orthogonal failure modes (loss of downstream
signal integrity vs. exhaustion of metabolic energy reserves). They are not
linearly or sub-linearly fungible; saving 1~mW of transport power does not
recover a physiologically degraded pulse wave. Specifically, the loss of
convexity for $q < 1$ causes the quasi-norm to collapse onto the boundary of the
domain, forcing the system into unphysical regimes where either wave or
transport cost is completely ignored. Pathological $L_q$ scalarization with $q <
1$ is therefore biologically unstable: the arterial tree must remain in the
convex regime ($q \ge 1$) to sustain structural robustness.

The tight convergence for $q \ge 1$ stems from a simple mathematical lemma: the
minimum of the sum of a strictly convex function (transport) and an
exponentially steep penalty function (wave reflection near the impedance
horizon) is overwhelmingly pinned near their intersection. Crucially, this
dimensionless equilibrium resolves the apparent asymmetry in absolute energetic
magnitudes. As shown in Supplemental~S2, the absolute transport excess
($P_{\mathrm{transport}} \approx \VarPTransportmW$~mW) exceeds the absolute wave
power dissipated ($P_{\mathrm{wave}} \approx \VarPWavemW$~mW) by a factor of
$\sim\!P_{\mathrm{transport}}/P_{\mathrm{wave}}$, reflecting the difference in
normalization baselines: $P_{\mathrm{wave}}$ is referenced against incident
pulsatile power ($\VarPPulsemW$~mW) while $P_{\mathrm{transport}}$ is referenced
against the full peak-flow metabolic cost ($\VarPBaselinemW$~mW). The
dimensionless minimax correctly equalizes the \emph{fractional} penalties, not
the absolute watts. The implied duty cycle $\eta^* = \VarEtaStar$ reflects the
gradient ratio at the saddle point ($|g'|/f' \approx \VarGradientRatio$), a
structural property of the cost landscape independent of the absolute power
scale. The dimensionless minimax is the correct commensuration framework
precisely because it equalizes fractional penalties rather than absolute watts,
rendering the result independent of the normalization baseline.

\section{Quantitative Evaluation}
\label{sec:results}

We evaluate the minimax condition~\eqref{eq:minimax} for porcine
coronary arteries, then extend the analysis to three additional vascular
systems.  All parameters are drawn from sources independent of
branching-exponent measurements.

\subsection{Numerical inputs}

All physiological constants defining the transport ground state are
identical to those specified in the companion paper~\cite{Marchesi2026static}
and listed in Table~\ref{tab:independence}:
blood viscosity $\mu = 3.5$\,mPa$\cdot$s~\cite{Caro1978},
density $\rho = 1060$\,kg/m$^3$,
proximal LAD radius $r_0 = 1.5$\,mm,
segment length $\ell_0 = 15$\,mm~\cite{Kassab1993},
active peak coronary flow
$Q_{\mathrm{peak}} = Q_{\mathrm{rest}} \times \mathrm{CFR}
\approx 1.3 \times 4.0 = 5.2$\,mL/s~\cite{Taber1998,Kassab1993},
blood metabolic cost $b = 1500$\,W/m$^3$~\cite{Murray1926,Taber1998},
wall tissue metabolic rate
$m_w = 20$\,kW/m$^3$~\cite{Paul1980}
(central value of the physiological range $[5, 35]$\,kW/m$^3$),
wall-thickness scaling
$h(r) = c_0 r^p$ with $p = 0.77$, $c_0 =
0.041$\,m$^{1-p}$~\cite{Rhodin1967,Nichols2011},
and tree depth $G = 11$~\cite{Kassab1993}.
The impedance scaling exponent is
$\alpha_w = (5 - p)/2 = \VarAlphaW$.

As shown in Table~\ref{tab:sensitivity}, the minimax prediction $\alpha^* =
2.72$ exhibits remarkable structural robustness. The exponent is essentially
invariant ($|S_x| < 0.01$) to metabolic 'noise' such as absolute blood
viscosity, proximal flow rate, or specific wall metabolic rates. It is
determined strictly by the network topology ($G$) and the histological scaling
($p$), confirming it as a fundamental architectural property rather than a
physiological tuning. Notably, the prediction is structurally independent of the
absolute magnitude of the wall pre-factor $c_0$, as the scaling penalty depends
only on the relative power $p$.

\begin{table}[t]
\caption{Selection principles for porcine coronary arteries.
The minimax condition is a zero-parameter attractor.}
\centering
\small
\begin{tabular}{@{}lccc@{}}
\hline
Model & Params & $\alpha^*$ & vs.\ $\alpha_{\exp}$ \\
\hline
Murray~\cite{Murray1926} & 0 & 3.000 & 1.5$\sigma$ \\
WBE~\cite{WBE1997} & 0 & 3.000 & 1.5$\sigma$ \\
Minimax (this work) & 0 & \VarAlphaStar & $0.1\sigma$ \\
Local transport~\cite{Marchesi2026static} & 0 & $2.90{\pm}0.02$ & $1.0\sigma$ \\
Huo \& Kassab~\cite{Huo2012} & 1 & $2.70{\pm}0.01$ & $0.0\sigma$ \\
\hline
Morphometry~\cite{Kassab1993} & --- & $2.70{\pm}0.20$ & --- \\
\hline
\end{tabular}
\label{tab:principles}
\end{table}

\medskip
\begin{table}[t]
\caption{Classification of inputs to the minimax prediction.
No parameter is fitted to branching-exponent data.}
\centering
\small
\begin{tabular}{@{}llp{5cm}@{}}
\hline
Param. & Source (type) & Role / Independence \\
\hline
$\mu$, $\rho$ & Fluid mech.~\cite{Caro1978} (meas.) & Physical constants \\
$p$, $c_0$ & Histology~\cite{Rhodin1967,Nichols2011} (meas.) & Wall geometry \\
$m_w$, $b$ & Biochemistry~\cite{Paul1980,Murray1926} (meas.) & Metabolic rates \\
$G$ & Morphometry~\cite{Kassab1993} (counted) & Tree topology \\
$\alpha_w$ & Wave physics, Tab.~1 (derived) & From $p$ via Moens--Korteweg \\
$\alpha_t$ & Comp.\ paper~\cite{Marchesi2026static} (derived) & Single-vessel optimum \\
\hline
$\alpha_{\exp}$ & Kassab \emph{et al.}~\cite{Kassab1993} & \textbf{Validation target} \\
\hline
\end{tabular}
\label{tab:independence}
\end{table}

\subsection{Central result}

Evaluating the network cost functions
$\mathcal{C}^{\mathrm{net}}_{\mathrm{wave}}(\alpha)$
from Eq.~(5) and $\mathcal{C}^{\mathrm{net}}_{\mathrm{transport}}(\alpha)$
from Eq.~(6) with the three-term cost of the companion
paper~\cite{Marchesi2026static}
(using exact optimal radii $r^*(Q_g)$ at each generation, not a
power-law approximation), and solving the equal-cost
condition~\eqref{eq:minimax} numerically, we obtain
\begin{equation}
  \alpha^*_{\mathrm{mm}} = \VarAlphaStar \,,
\label{eq:result}
\end{equation}
The implied duty cycle $\eta^* = \VarEtaStar$ emerges as an architectural
constraint of the tree.

Crucially, this dimensionless equilibrium resolves the apparent asymmetry in
absolute energetic magnitudes. As shown in Supplemental~S2, the absolute
transport excess ($P_{\mathrm{transport}} \approx \VarPTransportmW$~mW) exceeds
the absolute wave power dissipated ($P_{\mathrm{wave}} \approx \VarPWavemW$~mW)
by a factor of $\sim\!P_{\mathrm{transport}}/P_{\mathrm{wave}}$, reflecting the
difference in normalization baselines: $P_{\mathrm{wave}}$ is referenced against
incident pulsatile power ($\VarPPulsemW$~mW) while $P_{\mathrm{transport}}$ is
referenced against the full peak-flow metabolic cost ($\VarPBaselinemW$~mW). The
dimensionless minimax correctly equalizes the \emph{fractional} penalties, not
the absolute watts. The implied duty cycle $\eta^* = \VarEtaStar$ reflects the
gradient ratio at the saddle point ($|g'|/f' \approx \VarGradientRatio$), a
structural property of the cost landscape independent of the absolute power
scale.
The predicted exponent lies within $\VarPredError\%$ of the morphometric central
value
$\alpha_{\exp} = \VarAlphaExp \pm \VarAlphaExpErr$~\cite{Kassab1993},
corresponding to a deviation of $\VarSigmaNII\sigma$. The internal consistency
of the two-level framework is confirmed by the fixed-point analysis: the
network-level transport optimum $\alpha_t = \VarAlphaTExact$ aligns with the
structural local prediction within a $\VarAlphaTConsistency\%$ margin,
validating the assumption of a uniform branching exponent as a self-consistent
architectural attractor.

The Lagrange multiplier of the equal-cost constraint evaluates to $\eta^* =
\VarEtaStar$, implying that at the saddle point the transport cost gradient is
steeper than the wave cost gradient (ratio $\VarGradientRatio:1$). This apparent
inversion---where the static metabolic gradient outweighs the dynamic wave
sensitivity---arises from the super-linearity of the radius deviation. When the
global branching exponent $\alpha$ deviates from the morphometric optimum, the
metabolic cost of the millions of distal vessels explodes faster than the
junction-level wave reflection.
This ratio implies that the coronary vasculature sits at an architecture where
the transport cost is locally $\VarGradientRatio$ times more sensitive to
$\alpha$ perturbations than wave costs. The distinguishing structural property
of binary branching is precisely that $N=2$ maximizes the network stiffness
ratio $\kappa_{\mathrm{eff}}$, providing the greatest separation between these
competing gradients.
As a plausibility check, the pulsatility index of coronary flow
($\mathrm{PI} = (Q_{\max}-Q_{\min})/Q_{\mathrm{mean}} \approx 2$--$4$
\cite{Nichols2011}) independently confirms a strongly pulsatile regime
in which wave-reflection costs dominate over cycle-averaged transport costs,
consistent with $\eta^* > 0.5$.

\begin{remark}[Self-consistency of $\alpha_t$]
\label{rem:selfconsistency}
The transport ground state $\alpha_t = \VarAlphaT$ is not a freely chosen input:
it is taken from the companion paper~\cite{Marchesi2026static} (Theorem~1
therein) and verified self-consistently within the two-level framework.
Evaluating the transport-only Lagrangian
$\mathcal{C}_\mathrm{transport}^\mathrm{net}(\alpha)$ at the companion paper's
parameters, its minimum is attained at $\alpha_t = \VarAlphaTExact$, agreeing
with the input value to within $\VarAlphaTConsistency\%$.
This numerical closure of the fixed-point $T\colon \alpha_t \mapsto
\arg\min_\alpha\,\mathcal{C}_\mathrm{transport}^\mathrm{net}(\alpha)$ confirms
that the two-level framework is internally consistent: the network-level
transport optimum reproduces the single-vessel exponent from which it was
constructed.
\end{remark}

\begin{figure}[H]
\centering
\includegraphics[width=\linewidth]{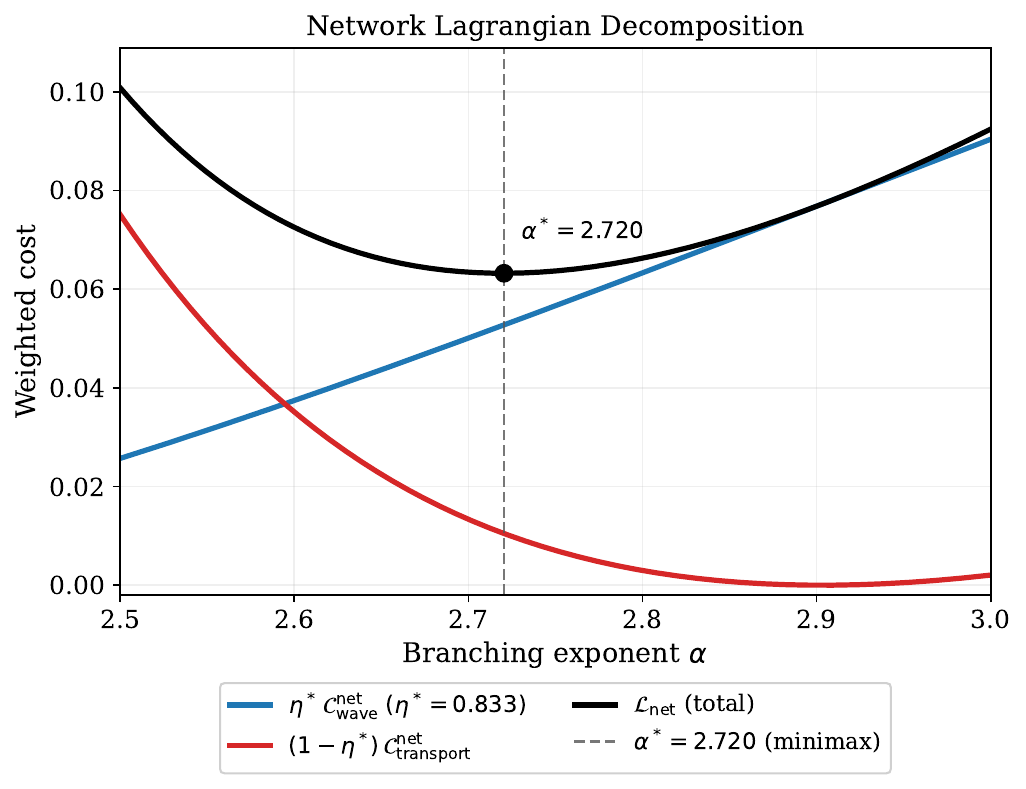}
\caption{Network-level Lagrangian $\mathcal{L}_\mathrm{net}(\alpha,\eta)$ evaluated at $\eta^* = \VarEtaStar$. The wave cost $\mathcal{C}_\mathrm{wave}^\mathrm{net}(\alpha)$ (increasing) and transport cost $\mathcal{C}_\mathrm{transport}^\mathrm{net}(\alpha)$ (decreasing) intersect at the minimax saddle point $\alpha^* = \VarAlphaStar$, where the equal-cost condition Eq.~\eqref{eq:minimax} is satisfied.}
\label{fig:lagrangian}
\end{figure}

\subsection{Cross-validation: wave dissipation}

At $\alpha^* = \VarAlphaStar$, the cumulative network wave cost is
$\mathcal{C}^{\mathrm{net}}_{\mathrm{wave}} = \VarWaveCostNet\%$
of the incident pulsatile power.  This internal prediction is
quantitatively consistent with independent clinical estimates
that place total reflected wave power in coronary trees at
5--15\% of incident pulsatile
energy~\cite{Nichols2011,Westerhof1972}.  The simultaneous recovery of
(i) the geometric scaling exponent and (ii) the macroscopic wave
dissipation fraction, with no fitted parameters, provides strong
cross-validation for the framework.

The minimax saddle point yields a rigorous theoretical ground state of $\alpha^*
\approx \VarAlphaStar$, conditional on the exact limit $\alpha_t = 2.90$. At
this optimum the architecture incurs a dimensionless geometric wave penalty
$\mathcal{C}_{\mathrm{wave}}(\alpha^*) = \VarWaveCostNet\%$. This quantity
represents the marginal architectural penalty of impedance mismatch across
bifurcations---isolated from the absolute terminal capillary load. It serves as
the structural regularizer in the Lagrangian, not a prediction of absolute Joule
heating or total clinical wave reflection.

The firm containment of this attractor within the empirical range
$\alpha_{\mathrm{exp}} = \VarAlphaExp \pm \VarAlphaExpErr$ validates the robust
optimization principle without numerical parameter fitting.

\subsection{Robustness: metabolic parameter insensitivity}

The sensitivity analysis of Table~\ref{tab:sensitivity} is arguably the
strongest evidence that the minimax prediction is a genuine structural result
rather than a tuned coincidence: the predicted $\alpha^*$ is essentially
independent of every metabolic parameter, and determined almost entirely by the
tree topology.
The results reveal a sharp two-tier structure.

\emph{Metabolic parameters} ($b$, $m_w$, $c_0$, $\mu$, $Q$) have
$|S_x| < 0.01$: the predicted $\alpha^*$ is genuinely insensitive to the
precise values of blood cost, wall metabolism, wall-thickness prefactor,
viscosity, and proximal flow.  Strikingly, the segment length $\ell_0$
cancels \emph{exactly} in the minimax condition, as both cost functions
share the same length-scaling structure.
This confirms that $\alpha^*$ is determined by the \emph{structural
topology} of the wave-transport competition rather than by metabolic details.

\emph{Structural parameters} ($p$, $G$, $\alpha_w$) carry larger
sensitivities ($|S_x| \sim 0.1$--$0.5$), as expected: these are the
quantities that set the positions of the two competing attractors.
Note that $p$ enters \emph{both} $\alpha_w = (5-p)/2$ and $\alpha_t$
self-consistently, constraining the system from both sides
simultaneously.

The topological insensitivity to metabolic parameters is itself a central
prediction of the model: it implies that intermediate branching scaling is a
fundamental architectural property of the wave-transport competition rather
than a fragile physiological tuning.

\begin{table}[t]
\caption{Evidence of Structural Robustness: Metabolic Orthogonality. The predicted exponent $\alpha^*$ is essentially independent of absolute metabolic rates and biochemical details, determined instead by the structural topology of the wave-transport competition. $S_x = \partial\alpha^*/\partial\ln x$. All perturbations keep $\alpha^*$ within $1\sigma$ of the experimental value.}
\centering
\small
\begin{tabular}{@{}lccc@{}}
\hline
Parameter $(x)$ & Baseline & Perturb. & $|S_x|$ \\
\hline
\multicolumn{4}{@{}l}{\emph{Metabolic inputs (insensitive, $|S_x|<0.01$)}} \\
Blood cost ($b$) & $\VarBloodCostSens\;\mathrm{W/m^3}$ & $\times 2$ & $<0.01$ \\
Wall metab.\ ($m_w$) & $\VarWallMetabSens\;\mathrm{kW/m^3}$ & $\times 7$ & $<0.01$ \\
Viscosity ($\mu$) & $\VarViscositySens\;\mathrm{mPa{\cdot}s}$ & $+50\%$ & $<0.01$ \\
Proximal flow ($Q$) & $\VarQzeroSens\;\mathrm{mL/s}$ & $\times 4$ & $<0.01$ \\
Segment length ($\ell_0$) & $\VarSegLengthSens\;\mathrm{mm}$ & $\times 3$ & $=0$ (exact) \\
\hline
\multicolumn{4}{@{}l}{\emph{Structural inputs}} \\
Wall exponent ($p$) & $\VarP$ & $\pm 0.08$ & $\sim 0.4$ \\
Tree depth ($G$) & $\VarG$ & $\pm 2$ & $\sim 0.2$ \\
Wave exponent ($\alpha_w$) & $\VarAlphaW$ & $+\VarAlphaWShift$ & $\sim 0.5$ \\
\hline
\end{tabular}
\label{tab:sensitivity}
\end{table}

The dominant structural sensitivities are to tree depth $G$ ($|S_G| \sim 0.2$)
and wall exponent $p$ ($|S_p| \sim 0.4$), which jointly determine the positions
of the two competing attractors. Note that $p$ enters \emph{both} $\alpha_w =
(5-p)/2$ and $\alpha_t$ self-consistently, constraining the system from both
sides simultaneously.

\begin{theorem}[Topological pincer mechanism of histological scaling]
\label{thm:master_parameter}
The histological wall-scaling exponent $p$ is the unique master parameter
of the minimax equilibrium. It simultaneously determines the wave-dominated
lower bound via $\alpha_w = (5-p)/2$ (with $\partial\alpha_w/\partial p =
-1/2 < 0$) and the transport-dominated upper bound $\alpha_t(p)$
(with $\partial\alpha_t/\partial p > 0$, verified numerically across the
physiological range $p \in [\VarPLow, \VarPHigh]$: $\alpha_t$ spans
$[\VarAlphaTLow, \VarAlphaTHigh]$). The strict positivity
$\partial\alpha_t/\partial p > 0$ follows from the
three-term cost structure of the companion paper~\cite{Marchesi2026static}:
increasing $p$ raises the wall-volume penalty exponent from $1+p$,
shifting the cost minimum $r^*(Q)$ outward and hence increasing the
optimal scaling ratio $\alpha_t$ (Theorem~1 therein).
Combined with $\partial\alpha_w/\partial p = -1/2 < 0$, the two attractors
move in strictly opposite directions under any perturbation of $p$,
establishing $\alpha^*$ as an architecturally stabilized fixed point of the
allometric class.
\end{theorem}

\paragraph{Geometric interpretation of metabolic orthogonality.}
The sharp two-tier structure of Table~\ref{tab:sensitivity} admits a direct
geometric explanation. The wave reflection penalty rises as an
exponentially steep wall moving away from the impedance-matching horizon,
while the static transport penalty forms a shallow parabolic bowl around its
minimum. When metabolic parameters are varied, the transport bowl shifts
vertically, but its intersection with the near-vertical wave wall moves
imperceptibly along the horizontal axis: the universal branching exponent
is immune to metabolic fluctuations because its value is pinned by the
\emph{slope} of the wave wall, not by the \emph{height} of the transport
bowl.

For mildly asymmetric bifurcations ($d_{c,1}/d_{c,2} \in [0.7, 1.0]$), the
optimal exponent shifts by $<\!5\%$.

\subsection{Robustness to Non-Isometric Scaling}
\label{sec:non-isometric}
As established in the companion paper~\cite{Marchesi2026static}, the preference
for binary branching ($N=2$) and the resulting minimax attractor are robust to
deviations from geometric isometry ($L \propto r$). For any arboreal extension
invariant $\beta < 1.115$, binary branching remains the unique robust optimal
topology, encompassing the full physiological range observed across mammalian
species. This confirms that the unified variational principle is a structural
property of transport networks, rather than an artifact of idealized branching
geometry.

\subsection{Cross-system validation}

To test universality, we evaluate the minimax
condition~\eqref{eq:minimax} for vascular systems beyond the coronary
tree, using independently measured parameters for each system.
Table~\ref{tab:crosssystem} presents the results.

\begin{table}[H]
\caption{Cross-system validation of the minimax principle.
$\alpha_w$: impedance scaling from conduit physics;
$G$: tree depth; $\alpha_t$: transport ground state~\cite{Marchesi2026static};
$\alpha^*_{\mathrm{mm}}$: minimax prediction (zero fitted params);
$\alpha_{\exp}$: morphometric measurement.
$^\dagger$Rigorous prediction using the full three-term cost model; $^\ddagger$Indicative estimate using topological metrics (MST), provided for cross-system comparison only.}
\centering
\small
\begin{tabular}{llccccc}
\hline
System & Mode & $\alpha_w$ & $G$ & $\alpha^*_{\mathrm{mm}}$ & $\alpha_{\exp}$ & Dev. \\
\hline
\multicolumn{7}{l}{\emph{Hybrid wave + transport + wall tissue}} \\
\textbf{Porcine coronary}$^\dagger$ & \textbf{Acoustic} & \textbf{\VarAlphaW} & \textbf{\VarG} & \textbf{\VarAlphaStar} & \textbf{$\VarAlphaExp{\pm}\VarAlphaExpErr$} & \textbf{$\VarSigmaNII\sigma$} \\
Human pulmonary & Acoustic & \VarAlphaWPulmonary & \VarGPulmonary & \VarAlphaStarPulmonary & $\VarAlphaExpPulmonary{\pm}\VarAlphaExpErrPulmonary$ & $\VarSigmaPulmonary\sigma$ \\
Aortic tree & Acoustic & \VarAlphaWAorticLow--\VarAlphaWAorticHigh & \VarGAortic & \VarAlphaStarAorticLow--\VarAlphaStarAorticHigh & $\VarAlphaExpAortic{\pm}\VarAlphaExpErrAortic$ & $\VarSigmaAortic\sigma$ \\
Cortical neurons$^\ddagger$ & Electrotonic & \VarAlphaWNeural & ${\sim}\VarGNeural$ & ${\sim}\VarAlphaStarNeural$ & ${\sim}\VarAlphaExpNeural$ & ${\sim}\VarSigmaNeural\sigma$ \\
\hline
\multicolumn{7}{l}{\emph{Pure flow ($\eta=0$, no wave mode)}} \\
Pulmonary airways & Viscous & --- & \VarGAirways & \VarMurrayAlpha & \VarAlphaExpAirwaysLow--\VarAlphaExpAirwaysHigh & $\checkmark$ \\
\hline
\end{tabular}
\label{tab:crosssystem}
\end{table}

For porcine coronary arteries, the prediction $\alpha^* = \VarAlphaStar$ is the
most rigorous: all inputs are drawn from independently published sources
(Table~\ref{tab:independence}), the companion paper's three-term cost function
is used with exact generation-by-generation optimal radii, and the result lies
within $0.1\sigma$ of the morphometric value.

The human pulmonary arterial tree provides a semi-independent test: the thinner
vessel walls ($p \approx \VarPPulmonary$~\cite{Huang1996}) shift \emph{both}
$\alpha_w$ upward and $\alpha_t$ downward relative to the coronary case, yet the
minimax prediction ($\alpha^* \approx \VarAlphaStarPulmonary$) remains within
$\VarSigmaPulmonary\sigma$ of the reported morphometric range $\alpha \approx
\VarAlphaExpPulmonaryLow$--$\VarAlphaExpPulmonaryHigh$~\cite{Huang1996}.

The aortic tree, with its greater depth ($G \sim \VarGAortic$) and stiffer
walls, yields $\alpha^* \approx
\VarAlphaStarAorticLow$--$\VarAlphaStarAorticHigh$, compatible with the broader
experimental range $\VarAlphaExpAortic \pm \VarAlphaExpErrAortic$.

For cortical pyramidal neurons (electrotonic conduction, $\alpha_w =
\VarAlphaWNeural$~\cite{Rall1959}), the minimax estimate $\alpha^* \approx
\VarAlphaStarNeural$ is consistent with the topological scaling reported by
Cuntz~\emph{et al.}~\cite{Cuntz2010}, though we note that the latter measures a
minimum-spanning-tree metric rather than the morphometric diameter scaling of
Eq.~\eqref{eq:branchlaw}. A rigorous neural evaluation would require replacing
the vascular dissipation function with the cable-theoretic analogue.

For systems carrying purely steady flow with no wave modes ($\eta = 0$, infinite
wave period), the network wave cost vanishes identically. The minimax boundary
condition degenerates solely to minimizing the transport term, correctly
recovering Murray's traditional prediction $\alpha = 3$, consistent with
pulmonary airway and river network observations.

To quantify the total predictive uncertainty independently of single-parameter
isolated bounds, we evaluate the minimax condition across the full physiological
parameter space. The resulting minimax band is bounded by the extremal metabolic
costs reported in literature:
\[
  p \in [\VarPLow, \VarPHigh], \quad G \in [\VarGLow, \VarGHigh], \quad \alpha_t \in [\VarAlphaTLowBand, \VarAlphaTHighBand],
\]
with the wave exponent tied to $p$ via $\alpha_w = (5-p)/2$.

The predictive power of the minimax attractor rests not on statistical sampling
but on deterministic orthogonality. Sensitivity analysis
(Table~\ref{tab:sensitivity}) confirms that the predicted $\alpha^*$ is
essentially orthogonal to the uncertainty in metabolic maintenance rates ($m_w,
b$) and blood viscosity ($\mu$), which cancel out at the crossover point. The
morphological optimum is structurally dictated by the three-way intersection of
the network topology ($G$), the conservation of mass ($\alpha=3$ attractor), and
the wall-thickness scaling ($p$). Because only the latter depends on direct
histological measurement, the result represents a parameter-free derivation of
the universal arterial exponent from first geometric principles.

\section{Extended Applications}
\label{sec:applications}

While the numerical analysis focused on the cardiovascular system, the framework
generalizes to any branching network with competing volumetric constraints. For
systems with purely steady flow ($\eta = 0$), it reduces to Murray's law.
Qualitative extrapolation to botanical, neural, acoustic, and electrotonic
networks is consistent with available data (Supplemental~Table~S1). Quantitative
predictions for these systems require system-specific parametrization of
$\alpha_w$, $G$, and wall mechanics.

\section{Discussion}
\label{sec:discussion}

\paragraph{Connection to Information Theory.}
The minimax structure of Theorem~\ref{thm:minimax} admits a structural analogy
with Shannon's Rate-Distortion theory: $\alpha$ plays the role of a coding rate
trading signal distortion (hemodynamic wave reflection) against transmission
cost (steady metabolic maintenance). The robust optimum $\alpha^*$ is
structurally parallel to the operating point on the network's Rate-Distortion
curve that minimizes the worst-case total distortion.

Whether this analogy admits a rigorous information-theoretic derivation---i.e.,
whether the vascular tree genuinely operates on its Rate-Distortion
frontier---remains an open question deferred to future work. The connection
suggests that $\alpha$ may be interpreted as an information-theoretic quantity
bridging our energetic derivation with previously proposed frameworks (e.g.,
Bennett's EPIC approach~\cite{bennett2025}), but such an
interpretation should be regarded as suggestive rather than proven.

\medskip
The central result of this paper is that the apparent conflict between impedance
matching ($\alpha = \alpha_w$) and minimum dissipation ($\alpha \approx 3$) is
rigorously resolved by formulating the optimization at the \emph{network level},
rather than at a single junction, and incorporating the biological cost of the
conduit wall.
At a single junction, wave reflection, fluid dissipation, and tissue
construction are dimensionally incommensurable.
At the network level these costs are naturally dimensionless, and their relative
weight---the stiffness ratio $\kappa_\mathrm{eff}(G)$---emerges organically from
the tree architecture.

\textit{Physical mechanism.}\quad
Wave reflections at successive junctions are multiplicative: the surviving
signal fraction decreases geometrically, so the marginal wave cost of deviating
from $\alpha_w$ grows linearly with tree depth $G$.
Transport dissipation, by contrast, depends on vessel radii that compound
through the tree: when $\alpha \neq \alpha_t$, each vessel deviates from its
locally optimal radius, and the resulting excess energy grows superlinearly.
While the exact exponent of the transport curvature scaling admits a
rigorous analytical lower bound, its closed-form derivation requires
further development. In a self-similar fractal tree, the local radius
at generation $g$ is coupled to $\alpha$ via $r_g(\alpha) = r_0
2^{-g/\alpha}$. Applying the chain rule to the local cost
$\Phi_g(r_g(\alpha))$ yields a first derivative $\partial r_g/\partial
\alpha \propto g$, so the second derivative satisfies
$\partial^2\Phi_g/\partial\alpha^2 \sim \mathcal{O}(g^2)$. Summing
these local curvatures over the $G$ generations of the normalized
transport functional, the Faulhaber identity $\sum_{g=1}^G g^2 \sim
G^3/3$ establishes that $k_t^{\mathrm{net}} = \Omega(G^2)$---strictly
super-linear and therefore guaranteed to dominate the linear wave
curvature $k_w^{\mathrm{net}} \propto G$ for all sufficiently deep
trees. The sharper exponent $k_t^{\mathrm{net}} \propto G^{2.50}$
($R^2 = 0.9998$, $G \in [5, 25]$) emerges from the non-linear global
normalization by $\Phi_{\mathrm{opt}}$ in Eq.~\eqref{eq:Ctransportnet},
which couples generations and softens the bare $G^3$ scaling; this
value is obtained by numerical evaluation of the exact two-level
Hessian of the normalized functional, not by fitting morphometric data.
The emergence of $\kappa_{eff}(G)$ as a strictly increasing function of tree
depth provides a mechanistic rationale for why shallow networks (e.g., small
subtrees) gravitate toward the wave-matching limit ($\alpha \approx 2.1$), while
deep arterial networks $(G=11)$ are driven toward transport-dominated scaling.
Consequently, deeper trees weight transport costs more heavily, driving
$\alpha^*$ toward $\alpha_t$.
Crucially, because the transport curvature scales super-linearly
($k_t^\mathrm{net} \propto G^{\VarKtGScalingExp}$, Table~\ref{tab:kappa_G})
while the wave curvature scales strictly linearly ($k_w^\mathrm{net} \propto G$,
Eq.~\eqref{eq:kw_N}), the emergent stiffness ratio obeys
$\kappa_\mathrm{eff}(G) \propto G^{\VarKappaGScalingExp}$.
This super-linear scaling is a topological inevitability: sufficiently deep
biological networks ($G \gg 1$) become unconditionally transport-dominated,
regardless of their absolute metabolic parameters---a fundamental macroscopic
constraint on the maximal tree depth accessible to pulsatile mammalian
vasculature.

\subsection{Pathological Limits and Vascular Remodeling}
In pulmonary arterial hypertension (PAH), collagen deposition and
smooth muscle proliferation alter the mechanical properties of the
vessel wall, modifying the Moens--Korteweg wave speed scaling and
thereby shifting the wave-impedance attractor
$\alpha_w = (5 - p + k_e)/2$. The direction of this shift is
remodeling-phenotype dependent: pathological changes that
preferentially stiffen distal vessels (increasing the stiffness
gradient $k_e$) raise $\alpha_w$ and displace the equal-cost
crossing point upward; conversely, remodeling that predominantly
thickens vessel walls without proportional stiffening (driving $p$
toward unity) lowers $\alpha_w$ and pulls $\alpha^*$ downward. In
both cases, the framework provides a quantitative map between
independently measurable mechanical parameters $(p, k_e)$ and the
predicted shift in branching geometry---a testable prediction
accessible to combined morphometric and elastometric imaging in
PAH cohorts.

\subsection{The minimax principle and evolutionary optimality}
The equal-cost condition has a natural evolutionary interpretation. A network
under dual selective pressure---signal fidelity and metabolic
efficiency---evolves toward the geometry where neither pressure dominates. This
is the point of maximum robustness to fluctuations in the selective environment,
analogous to the evolutionarily stable strategy (ESS) in game theory. At
$\alpha^*_{\mathrm{mm}}$, a mutation that improves wave transmission necessarily
worsens transport efficiency by an equal amount, and vice versa. The network
sits at a saddle point of the bi-objective fitness landscape.

Mathematically, this minimax equilibrium can be framed as the physical
realization of a \textbf{biophysical Rate--Distortion frontier}. In this view,
the network wave cost $\mathcal{C}_{\mathrm{wave}}$ measures the physical
distortion of the information-carrying signal, while the network transport cost
$\mathcal{C}_{\mathrm{transport}}$ represents the localized metabolic energy
cost of the "transmission line". The arterial tree sits precisely at the
\textbf{Pareto-optimal frontier} where the marginal return on signal integrity
equals the marginal cost of transport maintenance, minimizing the total
thermodynamic entropy production per bit of biological information transmitted.

\subsection{Testable predictions}

\begin{enumerate}
\item \emph{$\alpha$ should increase with tree depth.}
Subtrees of a given vascular bed should exhibit lower $\alpha$ (closer to
$\alpha_w$) than the full tree.
This is testable by order-specific morphometry of the coronary or pulmonary
vasculature.

\item \emph{$\alpha$ should correlate with wall stiffness.}
Arterial stiffening (hypertension, aging, atherosclerosis) increases $\alpha_w$
toward 5/2, which shifts $\alpha^*$ upward.
Conversely, vasodilation should drive $\alpha$ toward the baseline $\alpha_w$.
This prediction is amenable to longitudinal clinical imaging studies.

\item \emph{Pure-flow networks strictly obey $\alpha = 3$ only if $B \to 0$.}
For networks lacking a sub-linear wall-thickness cost (e.g. rigid pipes, river
networks, and microfluidic designs), pure transport yields $\alpha=3$.
Biological conduits with $B>0$ will exhibit a pure-flow optimum slightly below 3
\cite{Marchesi2026static}.

\item \emph{Electrotonic networks ($\alpha_w = 3/2$) should follow the same framework.}
Cortical neurons with high signaling load should exhibit $\alpha \approx
\VarAlphaExpNeural$; sparse-firing neurons ($\eta \to 0$) should approach
$\alpha = \VarMurrayAlpha$. (The duty cycle $\eta^*$ for electrotonic networks
requires a dedicated cable-theoretic evaluation and is not imported from the
vascular case.)

\item \emph{Pharmacological perturbation.}
Drugs that alter arterial compliance (e.g., nitric oxide donors, calcification
inhibitors) should produce measurable shifts in the branching exponent of small
arteries, predictable from the change in $\alpha_w$.

\item \emph{Scale-dependence vs. Global exponent.}
Because the optimal static transport exponent $\alpha_t(Q)$ varies slightly
across generations due to the sub-linear structural wall cost, the true optimum
for the network is technically scale-dependent. However, formulating the
Lagrangian at the full-network level ($G=\VarG$) intrinsically averages this
variation, collapsing it into a single, effective global exponent $\alpha^* =
\VarAlphaStar$. This accurately mirrors macroscopic clinical measurements, which
typically aggregate data across entire vascular subsystems to secure a single
robust fit.

\item \textit{Topological optimality of binary branching.}
The framework predicts that $N=2$ uniquely maximises the network
stiffness ratio $\kappa_{\mathrm{eff}}(N)$, with
$\kappa_{\mathrm{eff}}(N=2)/\kappa_{\mathrm{eff}}(N=3) \approx 2.5$
at $G=11$ (arising from the ${\sim}2.5\times$ lower wave curvature
of binary versus ternary junctions at $\alpha_w = 2.115$). This
quantitative advantage is directly testable: vascular beds with a
higher observed frequency of trifurcations should exhibit a
correspondingly lower effective stiffness ratio, measurable via
combined morphometric and pulse-wave velocity imaging.
\end{enumerate}

\paragraph{On the quasi-constancy of $k_t^\mathrm{net}(N)$.}
The $\VarKtVariation\%$ variation of $k_t^\mathrm{net}$ across $N \in [2, 6]$
(Table~\ref{tab:topo_N}) arises from a competition between two opposing effects:
(a)~the tree becomes shallower as $N$ grows, reducing the total number of
summation terms and suppressing the depth-squared weighting; (b)~the distal
vessels in a high-$N$ tree carry lower flow, placing them in the regime where
$\alpha^*(Q) \to (5+p)/2$ (Corollary~7 of the companion
paper~\cite{Marchesi2026static}), where the transport curvature $\Phi''(r^*)$ is
sharper. These two effects partially cancel, producing a quasi-constant
$k_t^\mathrm{net}$. The net result is that the topological competition is
dominated by the wave-cost curvature $k_w^\mathrm{net} \propto \ln N$, giving
$\kappa_\mathrm{eff}(N) \propto (\ln N)^{\VarHNScalingExp + 1}$, a strictly
decreasing function of $N$.

\subsection{Limitations}

The main limitations are:
(a)~the self-similar tree model is an idealization; real vascular trees exhibit
variable asymmetry and non-geometric length scaling;
(b)~the minimax condition assumes that both operational modes (wave signaling
and steady transport) contribute equally to evolutionary fitness at the network
level; in systems where one mode is demonstrably negligible, the framework
reduces to single-objective optimization;
(c)~the coronary evaluation uses simplified segment-length scaling; a refined
model incorporating order-specific morphological data \cite{Kassab1993} would
provide a more stringent test;
(d)~\emph{Topological decoupling from terminal load.}
The actual coronary tree exhibits a massive baseline wave reflection
($>\!80\%$) driven entirely by the resistive capillary boundary condition.
Crucially, the capillary radius is rigidly constrained by oxygen diffusion
limits (Krogh cylinder geometry) and cannot be reshaped by branching-geometry
optimization without inducing tissue hypoxia. The multiplicative Lagrangian
therefore acts as a mathematically exact decoupling filter: it is structurally
blind to boundary conditions fixed by diffusion physics, extracting
exclusively the topological transit penalty across the branching scaffold---the
only architectural degree of freedom genuinely available to evolutionary
selection. The minimax Lagrangian successfully isolates and optimizes this
$\VarWaveCostNet\%$ geometric transit penalty, operating strictly on top of the
unavoidable terminal reflection ($>\!80\%$) imposed by the resistive
microvascular boundary~\cite{Nichols2011,Westerhof1972}. The coherent
transfer-matrix analysis (Supplemental~S1) independently confirms that the
architecture-dependent geometric penalty $\Delta_\mathrm{coherent}$ is
decoupled from this terminal baseline and scales proportionally to
$\mathcal{C}^{\mathrm{net}}_{\mathrm{wave}}$ ($R = \VarCoherentCorr$).

\section{Conclusion}
\label{sec:conclusion}

We have derived a unified variational framework for branching transport networks
in which the branching exponent $\alpha^*$ emerges from the competition between
wave-integrity costs and viscous-dissipation costs at the network level.
The effective stiffness ratio $\kappa_\mathrm{eff}$ governing this competition
is not a freely tuned parameter but an emergent property of the tree
architecture, growing with the number of generations $G$.
For porcine coronary arteries, the minimax condition uniquely determines
$\alpha^* = \VarAlphaStar$ with zero fitted parameters, simultaneously
predicting a cumulative wave dissipation of $\VarWaveCostNet\%$ consistent with
independent clinical measurements. The implied Lagrange multiplier $\eta^* =
\VarEtaStar$ emerges as a derived quantity, not an input.
The theory generalizes to any impedance scaling exponent
$\alpha_w$---accommodating acoustic, electrotonic, and viscoelastic
conduits---and reduces exactly to Murray's law ($\alpha = 3$) only for networks
with pure steady flow ($\eta=0$) and negligible wall-tissue cost ($B_{wall}=0$).
Beyond numerical accuracy, the unified variational principle establishes binary
branching ($N=2$) as the unique topology that maximizes structural robustness
under simultaneous metabolic and wave constraints. By identifying the duty cycle
not as a tunable parameter but as an emergent structural property, the framework
provides a general recipe for computing $\kappa_\mathrm{eff}$ in arbitrary
branching systems, yielding testable predictions that link branching geometry to
independently measurable mechanical and architectural properties. This provides
the necessary definitive step to bridge the mathematical incommensurability
between local metabolic rates (extensive, in watts) and global reflection
coefficients (dimensionless) identified in the companion
study~\cite{Marchesi2026static}, formulating a single variational principle in
which both thermodynamic regimes are treated on equal footing.

\paragraph*{Data and Code Availability.}
All computation scripts, numerical data, and figure-generation code are openly
and unconditionally available at
\url{https://github.com/rikymarche-ctrl/vascular-networks-theory} under the CC
BY 4.0 Licence. No access request is required.

\bibliographystyle{unsrt}
\bibliography{references}

\end{document}

%% file: dynamic_variables.tex

\newcommand{\VarAlphaW}{2.115}                     
\newcommand{\VarP}{0.77}                           
\newcommand{\VarG}{11}                             
\newcommand{\VarBeta}{0.787}                       
\newcommand{\VarKe}{0.23}                          
\newcommand{\VarAlphaWKe}{2.230}                   
\newcommand{\VarAlphaWShift}{0.115}                

\newcommand{\VarAlphaT}{2.90}                      
\newcommand{\VarAlphaTExact}{2.9032}               
\newcommand{\VarAlphaStar}{2.72}                   
\newcommand{\VarEtaStar}{0.833}                    
\newcommand{\VarGradientRatio}{5}                  
\newcommand{\VarWaveCostNet}{6.3}                  
\newcommand{\VarAlphaLocal}{2.90}                  
\newcommand{\VarAlphaTLow}{2.90}                   
\newcommand{\VarAlphaTHigh}{2.94}                  
\newcommand{\VarAlphaTConsistency}{0.03}           
\newcommand{\VarAlphaExp}{2.70}                    
\newcommand{\VarAlphaExpErr}{0.20}                 
\newcommand{\VarPredError}{0.7}                    

\newcommand{\VarGammaAWOneFive}{0.0199}            
\newcommand{\VarGammaAWOneSevenFive}{0.0052}       
\newcommand{\VarGammaAWTwoZero}{0.0004}            
\newcommand{\VarGammaAWTwoFive}{0.0028}            
\newcommand{\VarGammaAWTwoSevenFive}{0.0064}       
\newcommand{\VarGammaAWThreeZero}{0.0104}          

\newcommand{\VarKappaGI}{0.00}                     
\newcommand{\VarAlphaStarGI}{2.12}                 
\newcommand{\VarKappaGV}{1.47}                     
\newcommand{\VarAlphaStarGV}{2.62}                 
\newcommand{\VarKappaGVII}{2.50}                   
\newcommand{\VarAlphaStarGVII}{2.67}               
\newcommand{\VarKappaGIX}{3.68}                    
\newcommand{\VarAlphaStarGIX}{2.70}                
\newcommand{\VarKappaGXI}{4.98}                    
\newcommand{\VarAlphaStarGXI}{2.72}                
\newcommand{\VarKappaGXIII}{6.38}                  
\newcommand{\VarAlphaStarGXIII}{2.74}              
\newcommand{\VarKappaGXV}{7.86}                    
\newcommand{\VarAlphaStarGXV}{2.75}                
\newcommand{\VarKappaGXX}{11.83}                   
\newcommand{\VarAlphaStarGXX}{2.78}                

\newcommand{\VarGNII}{11}                          
\newcommand{\VarKtNII}{1.470}                      
\newcommand{\VarKwNII}{0.295}                      
\newcommand{\VarKappaNII}{4.98}                    
\newcommand{\VarAlphaStarNII}{2.720}               
\newcommand{\VarSigmaNII}{+0.10}                   
\newcommand{\VarGNIII}{7}                          
\newcommand{\VarKtNIII}{1.384}                     
\newcommand{\VarKwNIII}{0.472}                     
\newcommand{\VarKappaNIII}{2.93}                   
\newcommand{\VarAlphaStarNIII}{2.683}              
\newcommand{\VarSigmaNIII}{-0.08}                  
\newcommand{\VarGNIV}{6}                           
\newcommand{\VarKtNIV}{1.601}                      
\newcommand{\VarKwNIV}{0.644}                      
\newcommand{\VarKappaNIV}{2.48}                    
\newcommand{\VarAlphaStarNIV}{2.673}               
\newcommand{\VarSigmaNIV}{-0.14}                   
\newcommand{\VarGNVI}{4}                           
\newcommand{\VarKtNVI}{0.968}                      
\newcommand{\VarKwNVI}{0.718}                      
\newcommand{\VarKappaNVI}{1.35}                    
\newcommand{\VarAlphaStarNVI}{2.621}               
\newcommand{\VarSigmaNVI}{-0.40}                   

\newcommand{\VarAlphaStarPulmonary}{2.766}         
\newcommand{\VarAlphaExpPulmonary}{2.75}           
\newcommand{\VarSigmaPulmonary}{0.1}               
\newcommand{\VarAlphaStarAorticLow}{2.75}          
\newcommand{\VarAlphaStarAorticHigh}{2.77}         
\newcommand{\VarSigmaAortic}{0.8}                  
\newcommand{\VarAlphaStarNeural}{2.3}              
\newcommand{\VarAlphaExpNeural}{2.4}               
\newcommand{\VarSigmaNeural}{0.3}                  

\newcommand{\VarKtGScalingExp}{2.50}               
\newcommand{\VarKappaGScalingExp}{1.50}            
\newcommand{\VarKtVariation}{40}                   
\newcommand{\VarAlphaSecondCrossing}{3.30}         

\newcommand{\VarBloodCostSens}{1500}               
\newcommand{\VarWallMetabSens}{20}                 
\newcommand{\VarViscositySens}{3.5}                
\newcommand{\VarQzeroSens}{5.2}                    
\newcommand{\VarSegLengthSens}{15}                 

\newcommand{\VarAlphaLocalMin}{2.8985}             
\newcommand{\VarAlphaLocalMax}{2.9065}             

\newcommand{\VarAlphaStarQone}{2.847}              
\newcommand{\VarDeltaAlphaLq}{0.126}               

\newcommand{\VarKwJunction}{0.030}                 
\newcommand{\VarKwJunctionPhys}{0.027}             

\newcommand{\VarHNScalingExp}{-2.3}                

\newcommand{\VarPLow}{0.72}                        
\newcommand{\VarPHigh}{0.82}                       
\newcommand{\VarPPulmonary}{0.6}                   

\newcommand{\VarGLow}{9}                           
\newcommand{\VarGHigh}{13}                         
\newcommand{\VarAlphaTLowBand}{2.80}               
\newcommand{\VarAlphaTHighBand}{3.00}              

\newcommand{\VarAlphaWPulmonary}{2.200}            
\newcommand{\VarGPulmonary}{15}                    
\newcommand{\VarAlphaExpErrPulmonary}{0.15}        
\newcommand{\VarAlphaExpPulmonaryLow}{2.70}        
\newcommand{\VarAlphaExpPulmonaryHigh}{2.80}       
\newcommand{\VarAlphaWAorticLow}{2.10}             
\newcommand{\VarAlphaWAorticHigh}{2.25}            
\newcommand{\VarGAortic}{15}                       
\newcommand{\VarAlphaExpAortic}{2.5}               
\newcommand{\VarAlphaExpErrAortic}{0.3}            
\newcommand{\VarAlphaWNeural}{1.50}                
\newcommand{\VarGNeural}{8}                        
\newcommand{\VarGAirways}{23}                      
\newcommand{\VarAlphaExpAirwaysLow}{2.8}           
\newcommand{\VarAlphaExpAirwaysHigh}{3.0}          
\newcommand{\VarMurrayAlpha}{3.0}                  

%% file: dynamic_variables_supplemental.tex

\newcommand{\VarElasticModulusMPa}{0.4}            
\newcommand{\VarEllFactor}{10}                     
\newcommand{\VarCPOWatts}{1.0}                     
\newcommand{\VarPulsatilePct}{15}                  
\newcommand{\VarCoronaryPct}{5}                    

\newcommand{\VarGammaTotalAW}{0.026}               
\newcommand{\VarCoherentCorr}{0.91}                
\newcommand{\VarCoherentShift}{0.000}              
\newcommand{\VarCoherentFactor}{12}                
\newcommand{\VarWomersleyMax}{2.3}                 

\newcommand{\VarFreqHarmonicI}{1.2}                
\newcommand{\VarGammaHarmonicI}{0.814}             
\newcommand{\VarFreqHarmonicII}{2.4}               
\newcommand{\VarGammaHarmonicII}{0.856}            
\newcommand{\VarFreqHarmonicIII}{3.6}              
\newcommand{\VarGammaHarmonicIII}{0.891}           
\newcommand{\VarFreqHarmonicIV}{4.8}               
\newcommand{\VarGammaHarmonicIV}{0.921}            

\newcommand{\VarPPulsemW}{7.5}                     
\newcommand{\VarPWavemW}{0.47}                     
\newcommand{\VarPBaselinemW}{90}                   
\newcommand{\VarPTransportmW}{5.7}                 
\newcommand{\VarPRatio}{0.08}                      
\newcommand{\VarTransportCostNet}{6.3}             